\documentclass[journal]{IEEEtran}

\IEEEoverridecommandlockouts                              
\include{psfig}
\usepackage{epsfig}
\usepackage{amsmath}
\usepackage{cases}
\usepackage{graphicx}
\usepackage{tikz}
\usepackage{multirow}
\usepackage[update,prepend]{epstopdf}
\usepackage{calc}

\usepackage{caption}
\usepackage{subcaption}

\def \l {\left(}
\def \r {\right)}

\def \be {\begin{equation}}
\def \ee {\end{equation}}
\def \ba {\begin{aligned}}
\def \ea {\end{aligned}}

\newcommand \edit[1] {{#1}}
\newcommand \editt[1] {{{#1}}}

\ifCLASSINFOpdf

\else
 
\fi

\title{The Economics of Quality Sponsored Data in Non-Neutral Networks}

\author{Mohammad Hassan Lotfi,
        Karthikeyan Sundaresan, Saswati Sarkar, and Mohammad Ali Khojastepour
\thanks{Mohammad Hassan Lotfi and Saswati Sarkar are affiliated with the Department of Electrical and System Engineering,\ at University of Pennsylvania, Philadelphia, PA. Karthikeyan Sundaresan and Mohammad Ali Khojastepour are affiliated with NEC Laboratories America, Princeton, NJ. Author's email addresses are \{lotfm,swati\}@ seas.upenn.edu, and \{karthiks,amir\}@ nec-labs.com.}
\thanks{Part of this work was presented in Wiopt'15 \cite{wiopt}.}
}

\begin{document}
\maketitle

\newtheorem{lemma}{\textbf{Lemma}}
\newtheorem{note}{\textbf{Note}}
\newtheorem{property}{\textbf{Property}}
\newtheorem{theorem}{\textbf{Theorem}}
\newtheorem{definition}{\textbf{Definition}}
\newtheorem{corollary}{\textbf{Corollary}}
\newtheorem{remark}{\textbf{Remark}}

\begin{abstract}
The growing demand for data has driven the Service Providers (SPs) to provide differential treatment of traffic to generate additional revenue streams from Content Providers (CPs). While SPs currently only provide best-effort services to their CPs, it is plausible to envision a model in near future, where CPs are willing to sponsor quality of service for their content in exchange of sharing a portion of their profit with SPs. This quality sponsoring becomes invaluable especially when the available resources are scarce such as in wireless networks, and can be accommodated in a non-neutral network. In this paper, we consider the problem of {\em Quality-Sponsored Data} (QSD) in a non-neutral network. In our model, SPs allow CPs to sponsor a portion of their resources, and price it appropriately to maximize their payoff. The payoff of the SP depends on the monetary revenue and the satisfaction of end-users both for the non-sponsored and sponsored content, while CPs generate revenue through advertisement. \editt{Note that in this setting, end-users still pay for the data they use.} We analyze the market dynamics and equilibria in two different frameworks, i.e. sequential and bargaining game frameworks, and provide strategies for (i) SPs: to determine if and how to price resources, and (ii) CPs: to determine if and what quality to sponsor. The frameworks characterize different sets  of equilibrium strategies and market outcomes depending on the parameters of the market.
\end{abstract}

\begin{IEEEkeywords}
net neutrality, game theory, sponsored data plans, wireless networks
\end{IEEEkeywords}

\maketitle \thispagestyle{empty}

\section{Introduction}

The growing demand for data and the saturating revenue of \edit{Broadband
Service Providers (SPs)} have driven the SPs to provide differential treatment of traffic to generate additional revenue streams from content providers (CPs). This occurs both in wired and wireless networks. In addition to the SPs generating revenue from the CPs, with such a model, CPs can ensure the quality of service they provide for their end-users particularly when resources are scarce such as in wireless networks.  

This differential treatment of traffic can be accommodated in a non-neutral network\footnote{\editt{An instance of a differential treatment in wireless cellular networks is AT\&T that has launched a \emph{sponsored data plan}~\cite{Att} that allows CPs to pay for the data bytes that their users consume, thereby not eating into the users' data quota. Thus, the content of the CPs that pay the sponsoring fee is treated differently than the content of the rest of the CPs.}}. This has raised serious concerns among net neutrality advocates, especially with the recent landmark ruling favoring Verizon in its case against the Federal Communication Committee (FCC) in January 2014~\cite{NetNeutrality}. In February 2015, FCC reclassified the Internet as a utility \cite{FCCutility}, providing the ground for this agency to secure even more strict net-neutrality rules. However, this will not be the end of the net-neutrality debates in U.S. Further actions, from the SPs and proponents of relaxing net-neutrality rules, are expected.  \edit{In addition, the net-neutrality issue is not restricted to the U.S. Many countries are debating on these rules, and have different policies for SPs to handle the traffic. For example, in October 2015, the European parliament has rejected legal amendments for strict net-neutrality rules, and passed a set of rules that allow for  sponsored data plans and Internet fast lanes for ``specialized services"\cite{EU_net_neutrality}.}

Given the incentives of SPs and some CPs to adopt a non-neutral regime, while SPs only provide best-effort service to their CPs in the current model, it is easy to envision a model in the near future, in which CPs require {\em quality of service} for the data they sponsor. For example, if YouTube wants to increase the number of its active users through sponsoring its videos, it would derive value (utility) from the sponsorship only if the videos are delivered at a good quality. 
We refer to this model as the {\em Quality-Sponsored Data} (QSD) model, wherein (spectral) resources at the SP are sponsored to ensure quality for the data bytes being delivered to the end users. This a significant departure from the current model in which CPs sponsor only the data bytes without considering any associated quality. Thus, there is no direct coupling between the scarce (wireless) resources at the SP and the market dynamics between the CP, SP and end-users. In contrast, the QSD model brings this coupling to the fore-front. 

Hence, the over-arching goal of this work is to analyze and understand the implications of the QSD model on the market dynamics, which we believe is both timely and important. Using game-theoretic \cite{MWG} tools, we study the market equilibria and dynamics under various scenarios and assumptions involving the three key players of the market, namely the CPs, SPs and end users. We investigate the scenarios under which the QSD model is plausible, and one can expect a stable outcome for the market that involves sponsoring \edit{the quality of the content} by CPs. In addition, we discuss about the division of profit between SPs and CPs in two cases (1) when the decision makers do not cooperate and at least one of them is myopic optimizer 
, and (2) when both cooperatively maximize the payoff in the long-run. In the process, we devise strategies for the CPs (respectively, SPs) to determine if they should participate in QSD, what quality to sponsor, and how the SPs should  price their resources. 

In our model, SPs make a portion of their resources available for sponsorship, and price it appropriately to maximize their payoff. Their payoff depends on monetary revenue and satisfaction of end-users both for the non-sponsored and sponsored content. Note that the QSD model couples market decisions to the scarce (wireless) resources. Thus, resources allocated to sponsored contents will affect those allocated to non-sponsored content and hence their quality. Thus, one should consider the impact of the quality of the two types of data (sponsored and non-sponsored) on satisfaction of end-users.

We consider one CP and one SP. We consider that the CP has an advertisement revenue model\footnote{A CP that earns money through advertisements.}, (Section~\ref{section:model}), and characterize the myopic pricing strategies for the CP and  the SP given the quality of the content that needs to be guaranteed and the available resource using a non-cooperative sequential game framework (Section~\ref{section:analysisI}). Assuming the demand for content to be dynamic, wherein the change in the demand is dependent on the quality end-users experience, we investigate the asymptotic behavior of the market when at most one of the decision makers (SP or CP) is short-sighted, i.e. not involving the dynamics of demand in their decision making. We show that depending on certain key parameters, such as the importance of non-sponsored data for SPs and the parameters of the dynamic demand, the market can be asymptotically (in long run) stable or unstable. Furthermore, four different stable outcomes are possible: 1. no-sponsoring, 2. maximum bit sponsoring: the CP sponsors all the available resources, 3. minimum quality sponsoring: the CP sponsors minimum resources to deliver a minimum desired rate to her users, and 4. Interior solution in which the CP sponsors more than the minimum but not all the available resources. We characterize the conditions under which each of these asymptotic outcomes is plausible. \edit{In Section~\ref{section:simulation_seq}, the effects of different market parameters on the asymptotic outcome of the market is investigated through numerical simulations.} Note that there may exist multiple equilibria, and a non-cooperative framework may lead to a Pareto-inefficient outcome. Thus, when both of the decision makers are long-sighted, it is natural to consider a cooperative scheme such as a bargaining game framework. Thus, in Section~\ref{section:bargaining}, we investigate the role of a CP and an SP with long-sighted business models\footnote{in which decision makers maximize their payoff in long-run considering the dynamics of the demand for the content.} in stabilizing the market and equilibrium selection. We characterize the Nash Bargaining Solution (NBS) of the game to determine the profit sharing mechanism between the SP and CP. \edit{We present the numerical results in Section~\ref{section:bargaining_simulation}. Finally, in Section~\ref{section:discussion_final}, we summarize the key results of the paper and comment on some of the assumptions and their generalizations.  }

\textbf{Related Works:}

New pricing schemes in the Internet market either target end-users or CPs. For the end-user side, different pricing schemes have been proposed to replace the traditional flat rate pricing  \cite{soumya_survey}, \cite{Tube}, \cite{ISIT}. These schemes can create additional revenue for SPs and provide a more flexible data plan for end-users. However, SPs are reluctant in adopting such pricing schemes due to the fact that these schemes are typically not user-friendly. Thus, SPs mainly focus on changing the pricing structure of the CP side, for which they should deal with net-neutrality rules.   

Works related to the emerging subject of sponsored content are scarce. In \cite{Mung}, \cite{andrews1}, \cite{andrews2}, and \cite{zhang}, authors investigate the economic aspects of content sponsoring in a framework similar to At\&t sponsored data plans. \editt{Note that in At\&t sponsored data plan, the CP pays for the quantity of the data carried to the end-users, while in our scheme the CP pays for the \emph{quality} of the data, and end-user is responsible for paying for the quantity.} We take into the account the quality of the content and the coupling it has with scarce  resources.  We consider more strategic CPs that decide on the portion of SP's resources they want to sponsor, based on the price SPs quote and the demand from end-users.  

This work falls in the category of economic models for a non-neutral Internet \cite{chiang,economides,asu,NetEcon,Katz,CISS},\editt{\cite{Kramer, Misra}}. A survey of the existing literature on the economics analysis of non-neutral markets is presented in \cite{survey}. Most of the works in this area study the social welfare of the market under neutrality and non-neutrality regimes (which includes [14-21]).  \edit{In these works the decision of CPs does not depend on the demand for the content, and simply is a take-it-or-leave-it choice, i.e. either the CP pays for the premium quality or uses the free quality. In addition, \editt{most of the works} do not consider the coupling between limited resources available to SPs and the strategies of the decision makers. \editt{Exceptions are \cite{Kramer} and \cite{Misra}.} We consider that CPs decide on the number of resources they want to sponsor based on the dynamics of their demand. Depending on the demand and number of resources available with the SP, the number of sponsored resources by the CP yields a quality of experience for users of sponsored and non-sponsored contents. Thus, we consider the coupling between market decisions and the limited wireless resources. Moreover, we study problems like stability of the market and the effects of being short-sighted or long-sighted.} Therefore, we focus on one-to-one interaction between CPs and SPs, and its implications on the payoff of individual decision makers.


The closest work to ours is \cite{richard} in which authors study the interaction between an SP and a CP when the CP can sponsor a quality higher than the minimum quality under a private contract with the SP. Their main focus is to compare the social welfare of the sequential game when either the SP or the CP is the leader, with the Pareto optimal outcome resulting from a bargaining game between the SP and the CP. Authors assume that the number of subscribers to the SP is an \emph{increasing} function of the quality it provides for the CP. In other words, as the quality for the sponsored content enhances, end-users of the SP become more satisfied. However, in our work, the main focus is the coupling between the limited resources and the quality. Thus, providing a better quality for a sponsored content may degrade the quality of non-sponsored contents in peak congestion times. Therefore, in our model, the satisfaction of end-users which is a function of both sponsored and the non-sponsored content is not necessarily increasing with respect to the sponsored quality. This changes the nature of the problem.

\section{Model}\label{section:model}

\subsection{Problem Formulation:}
We model the ecosystem as a market consisting of three players: CPs, SPs, and end-users. A summary of important symbols is presented in Table~\ref{table:symbols}. We focus on the interaction between SPs and CPs, and not on the competition among SPs and CPs. Thus the interaction between one CP and one SP is considered. The strategy for the CP is to determine how much resources to sponsor (i.e. quality), and the strategy of the SP is to determine how to price her resources. Decisions are made by the players  \edit{based on an estimated demand update function (explained later)} at the beginning of every time-epoch, which captures the typical time granularity of sponsorship decisions\footnote{\edit{Using the estimate of the demand, players decide on their strategy to maximize  an ``estimated" payoff (and not the actual one). Note that the shorter the interval of epochs, the more accurate the estimates, and the more inconvenient the implementation would be. \editt{We will observe that, in our framework, the algorithm of decision making in NE would be simple. Thus, the decision making can be done in shorter time intervals, e.g. minutes.}}}. 
\begin{figure}[t]
    \centering
    \includegraphics[width=0.45\textwidth,height=2cm]{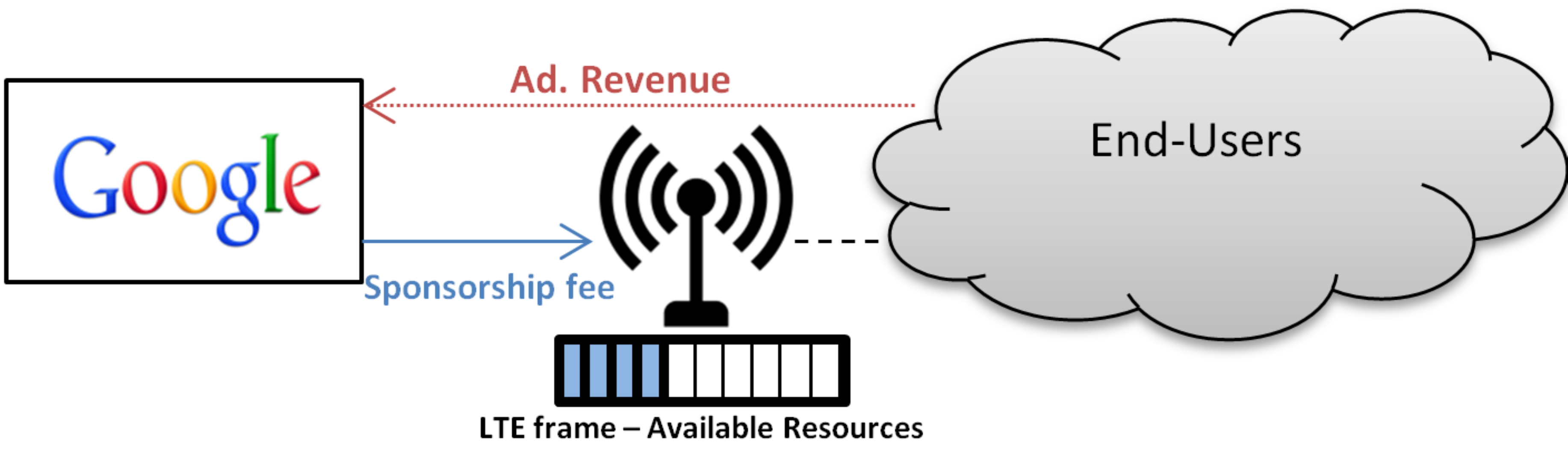}
    \caption{Market when CP has an advertisement revenue model}\vspace{-4mm}
    \label{figure:modelad}
\end{figure}  



The CP  has an advertisement revenue model, and sponsors $b_t$  resources (e.g. bits in an LTE frame) out of a total of $N$ resources at $t^{th}$ time-epoch to sponsor the \editt{average} quality  of at least $\zeta$ ($\frac{\text{bit}}{\text{frame}}$) for her content, and pays a price of $p_t$ per resource sponsored to the SP. \editt{Thus, on average a quality of $\zeta$ should be satisfied for the users. If not, the CP exits the sponsorship program. Note that this does not guarantee the quality of an individual user to be higher than $\zeta$.} An example Schematic picture of the market in this case is presented in Figure~\ref{figure:modelad}. 


The CP and the SP choose their strategy at time-epoch $t$ after observing the previous demand, i.e. \editt{the number of end-users desiring content from the previous epoch}. Obviously, the demand is non-negative. Note that the demand for the content of the CP changes over time depending on the satisfaction of users, which in turn depends on the resources that the CP decides to sponsor and hence the quality. \edit{We suppose that the satisfaction of users depends on the average quality, i.e. $\frac{b_t}{d_t}$}\footnote{\edit{Note that we are analyzing the model from the perspective of the CP and the SP. 	Thus, we are assuming that the CP and the SP see the demand for the content as a whole and want to sponsor an average sponsored quality. The demand of the individual end-users could be potentially different from each other. }}, and  the demand for content updates as follows\footnote{\editt{Note that receiving a satisfactory quality, increases the chance of user repeating the visit to the website and increases the number of new users that are going to use the service. Therefore, a satisfactory QoS will likely increase the demand for the data in the next session. In addition, we assume that the increase in the demand would be slower with high rates (a diminishing return behavior)}. },
\begin{eqnarray}\label{equ:d_t+1}
\small
d_{t+1} =
\left\{
\begin{array}{ll}
d_t \l 1+\gamma \log\l\kappa_u \frac{b_t}{d_t}\r\r^+  & \mbox{if }  d_t>0 \\
0 & \mbox{if } d_t=0
\end{array}
\right.
\end{eqnarray}  
\normalsize
where $z^+=\max\{z,0\}$, $d_t$ is the demand between epoch $t$ and $t+1$, $\frac{b_t}{d_t}$ is the rate a single user receives, and $\log\l\kappa_u \frac{b_t}{d_t}\r$ models the satisfaction of end users: the higher the rate received by users, the higher their satisfaction. 
The parameter $\gamma>0$ represents the \emph{sensitivity} of end-users to their satisfaction. A higher $\gamma$ is associated with a higher rate of change with respect to the satisfaction of users (higher fluctuation in demand). An instance of this type of users is customers of a streaming website like Netflix that are sensitive to the quality they receive. Parameter $\kappa_u>0$ is a constant.

Note that the total available wireless resources (for sponsored and non-sponsored contents) is limited ($N$). This limits the number of sponsored resources ($b_t$) which in term determines the upperbound of resources that can be allocated to non-sponsored contents. This is a key distinction of our work from previous works as the limited resources available couples the utility of end users for both sponsored and non-sponsored content with the decisions of the  market players. We assume that the  number of resources (bits) a CP can sponsor is bounded above by $\hat{N}$ ($\hat{N}\le N$).  

The utility of the CP if she chooses to enter the sponsorship program consists of the utility she receives by sponsoring the content minus the price she pays for sponsoring the sponsored bits. The latter is $p_t b_t$. The former, i.e. the utility of the CP for sponsoring the content, depends on the advertisement revenue which in turns depends on the demand for the content as well as the quality received by the users (throughput is $\frac{b_t}{d_t}$) \footnote{\edit{Note that $\frac{b_t}{d_t}$ can be the quality of the content, ads, or both. One example of CPs whose revenue depends on the quality of the ad is a CP that support video ads, e.g. YouTube.  Another example is websites loaded with several ``flash ads" for which users may have difficulty loading the ads which can lead to the decrease of number of clicks on the ads. In addition to these CPs,  we can think about scenarios in which increasing the quality of the content of a website (not only the ads) increases the revenue of this website. An example of such contents are shopping websites (e.g. Amazon). Improving the quality of the experience of users, increases the chance of spending more time on these website. This would increase the chance of a transaction which increases the revenue of the CP.}}. We consider the utility from advertisement for the CP to be:
\begin{eqnarray}
\small
u_{ad}(b_t) =
\left\{
	\begin{array}{ll}
		 \alpha d_t \log\l\frac{\kappa_{CP} b_t }{d_{t}} \r    & \mbox{if }  d_t>0 \\
		0 & \mbox{if } d_t=0
	\end{array}
\right.
\end{eqnarray}  
\normalsize

Note that the better the quality of advertisement, more successful the advertisement would be, and therefore the higher the utility that the CP receives from advertisement. Thus, the utility of advertisement is dependent on the satisfaction of users. This is the reason that we use a similar function  to \eqref{equ:d_t+1} for the utility of advertisement\footnote{\editt{Note that we expect a diminishing return of ad utility based on quality, i.e. after a certain point increasing the quality would not significantly increase the utility of advertisement. Thus, we used a log function to model the ad utility from an end-user ($\log(\frac{b_t}{d_t})$). Thus, we assume the utility of advertisement to be $\sum_{d_t}\text{constant}\times\log(\frac{b_t}{d_t})=\text{constant}\times d_t\log(\frac{b_t}{d_t})$. If we consider a linear dependency between quality and ad revenue, then the utility would be $\sum_{d_t}\text{constant}\times \frac{b_t}{d_t}=\text{constant}\times b_t$. However, we believe that a function with diminishing return would model the ad utility more closely.  }}. The constant $\kappa_{CP}$ in general can be different from $\kappa_u$. The parameter $\alpha$ \editt{is a constant that models the the unit income of the CP for each end-user based on the quality that the end-user receives}: The higher $\alpha$, the higher the profit of the CP per rates sponsored. An example of  CPs with high $\alpha$ is shopping websites  (e.g. Amazon) that in contrast with streaming websites (e.g. Netflix) have a high profit per user rate.  

Thus, the utility of the CP at time $t$ if she chooses to join the sponsorship program is:
\small
\begin{equation}\label{equation:utilityCP}
\begin{aligned}
u_{CP,t}(b_t)&=u_{ad}(b_t)-p_t b_t
\end{aligned}
\end{equation}
\normalsize
To have a non-trivial problem, we assume $\kappa_{CP} \zeta>e=2.72$. 
 
The utility of the SP at time $t$ if she chooses to offer the sponsorship program is the revenue she makes by sponsoring the bits plus the users' satisfaction function:

\small
\begin{equation}\label{equ:SPutility}
u_{SP,t}(p_t) =  p_t b_t+ u_s(b_t(p_t))
\end{equation}
\normalsize 
where the users' satisfaction function, i.e. $u_s(.)$, is a function of the number of sponsored bits which subsequently depends on the price $p_t$. This function consists of two parts: (i) the satisfaction of users for access to the sponsored content and its quality, and (ii) the satisfaction of users when using non-sponsored content. This function could be decreasing or increasing depending on the users, the cell condition, etc. We define the satisfaction function as follows\footnote{\editt{Note that in the case of no sponsoring, the demand of the CP should be added to the demand of in the best effort categor, i.e.  $D$  in \eqref{equ:old_B} should be substituted by  $d_t+D$. However, we assume that $d_t<<D$, i.e. the demand for one content is much smaller than the aggregate demand for all other contents. This often arises in practice.  In Appendix~\ref{appendix:approx}, we show that this modification does not alter any results in essence, and the insights of the model would be the same as before.}}

\footnotesize
\begin{subnumcases}{u_s(b_t)=}
  \nu_1   d_{t} \log\l \frac{\kappa_{SP} b_t}{d_{t}}\r+\nu_2 D \log \l \kappa_{SP} \frac{N-b_t}{D}\r \label{equ:old_A} \\
   \qquad \qquad \qquad \qquad \qquad \qquad \qquad  \ \  \mbox{if }  d_t>0 \& \ b_t>0 \nonumber \\
 	\nu_2 D \log \l \kappa_{SP} \frac{N-b_t}{D}\r  \quad \mbox{Otherwise} \label{equ:old_B}
\end{subnumcases}
\normalsize
where $D$ is \editt{the total demand for all CPs other than the strategic CP}\footnote{\editt{We now argue why $D$ is considered to be constant. We consider the content of the CP that is willing to sponsor her content to be different from the content of other CPs, i.e. no competition over the content. An example of such CP is Youtube (for personal video streaming). This yields that the demands for the strategic CP (that can be sponsored) and other CPs are independent of each other. Thus, no demand will be switched from the content of the sponsored  CP  to other CPs. In addition, since we  assume that other CPs are not sensitive to the quality they deliver,  their demand is independent of the quality their end-users receive. Thus, $D$ can be considered as a constant, and independent of the demand for the sponsored content.}} and $\kappa_{SP}$ is a positive constant. \edit{In addition, $v_1$ and $v_2$ are constants corresponding to the weights that end-users assign to the sponsored and non-sponsored data, respectively. We considered the users' satisfaction function to be a part of the SP's utility since it is natural to think that SPs not only  care about the money they receive for the sponsored content, but also about the satisfaction (or the overall quality of experience) of end-users for both sponsored and non-sponsored content\footnote{\editt{Note that in reality, end-users can switch between SPs if they are not satisfied, and this incurs loss to the SP they leave. To capture this, we need to consider the competition between SPs which makes the analysis much more complicated. Instead, we consider only one SP and the user satisfaction function to be an element in the utility of the SP. This captures some aspects of competition over end-users between SPs, without complicating the analysis unnecessarily.}}. Another reason for considering the satisfaction of users  is the regulatory policies on the quality of the experience of users who use non-sponsored contents. Thus, $v_1$ and $v_2$ can be determined by the SP or the regulator\footnote{\editt{ Although in the current set up, the regulator does not provide quality constraints for the SP, one can envision that in a non-neutral framework, the regulator imposes  explicit or implicit  constraints on the behavior of SP toward the sponsored and non-sponsored data. In other words, in a non-neutral regime, it is natural to suppose that the regulator forces the SPs to take into the account the satisfaction of their users regardless of the fact that they are using sponsored or non-sponsored data.  Thus, the SP wants to maximize her utility (which depends on the money collected from the CPs) given some constraints. In this sense, including the satisfaction of users with parameters $v_1$ and $v_2$ is similar to the \emph{Lagrangian penalty (reward) function} by which  we solve the mentioned maximization.  Note that eventually $v_1$ and $v_2$ is set by the SP and not the regulator. However, their value depends on the restriction determined by the regulator. Therefore, a strict net-neutrality rule, mandates the SP to  assign high weights to the quality of the  content of non-sponsored data, i.e. high $v_2$.}}. The higher $\frac{v_2}{v_1}$, the higher would be the importance of the non-sponsored content.}
	 
Note that, despite the dependencies between $\kappa_u$, $\kappa_{CP}$, and $\kappa_{SP}$, these parameters could be potentially different. \editt{For example, the ad revenue is paid by an advertiser. This advertiser may value the quality of the content delivered to the end-users different from the end-users. Thus, for  this reason, $\kappa_{CP}$ might be different from $\kappa_u$. We  do not mandate the parameters to be different from each other and they can be potentially equal.}. 

Recall that we assume if $d_t=0$ or one of the decision makers exits the sponsoring program, then the game ends, and we have a stable outcome of no-sponsoring.

\begin{table}[t]
\centering
    \begin{tabular}{ | p{2cm} || p{6cm} |}
    \hline
    Symbol & Description \\ \hline \hline
    $p_t$  & the price per unit of resources sponsored at time $t$ \\ \hline 
    $b_t$ & the number of sponsored bits in an LTE frame at time $t$ \\ \hline
    $d_t$ & the demand between epoch $t$ and $t+1$\\ \hline
$\zeta$ & the minimum average quality desired by end-users \\ \hline    
    $\gamma$ & sensitivity of end users to the quality they receive.\\ \hline
    $\alpha$ &  \editt{constant, the unit income} \\ \hline
    $\kappa_u$, $\kappa_{CP}$, $\kappa_{SP}$ & constants  \\ \hline
    $\hat{N}$ & the number of available bits for sponsoring \\ \hline
    $N$ & the total number of bits (resources) in an LTE frame\\ \hline
    $u_s(.)$ &  end-users' satisfaction function \\ \hline
    $u_{ad}(.)$ & CP's advertisement profit \\ \hline 
    $\nu_1$ & the weight end-users assign to the sponsored data \\ \hline
    $\nu_2$ & the weight end-users assign to the non-sponsored data\\ \hline
    $D$ & the total demand of end-users for non-sponsored data\\ \hline
    $\frac{1}{\kappa_u}$ & the stable quality, the rate that stabilizes the demand \\ \hline
	$z\ \& y\ $ & the participation factor for the CP and SP, $1=$join, $0=$exit
	\\ \hline

          \end{tabular}
    \caption{Important Symbols}\label{table:symbols}\vspace{-4mm}
\end{table}

\subsection{Preliminary Notations and Definitions:}

In this section, we define some notations that we use throughout the paper. Section-specific notations and definitions are presented in the corresponding sections.

Note that, we model the problem of QSD as a sequential game if at least one of the decision makers is short-sighted, and as a bargaining game when both CPs and SPs are long-sighted.

\begin{definition}\emph{Short-Sighted (Myopic) Decision Maker:}
A decision maker is short-sighted if she maximizes the myopic payoff \edit{knowing  the present demand ($d_t$)},\footnote{\edit{Mathematical definitions for the optimization solved by the short-sighted SP and CP are presented in Equations \eqref{equ:max_SP} and \eqref{equ:optimization_CP_1}, respectively.}} i.e. does not involve the evolution of demand \edit{\eqref{equ:d_t+1}} in her decision making.
\end{definition}

\begin{definition}\emph{Long-Sighted Decision Maker:}
A decision maker is long-sighted if she maximizes her payoff in long-run considering the current demand and the evolution of the demand \edit{in \eqref{equ:d_t+1}}.\footnote{\edit{Mathematical definitions for the payoff of  the long-sighted SP and CP are presented in Equations \eqref{equ:long_SP_max} and \eqref{equ:long_CP_max}, respectively.}}
\end{definition}

Since we consider an evolving demand of end-users based on their satisfaction, one of the contributions of this paper is to characterize the stability conditions of the market.

\begin{definition} \emph{Stable Market:} We say that the market is stable if and only if the demand of end-users is asymptotically stable, i.e. if and only if:
$$
\lim_{t\rightarrow \infty}{|d_{t+1}-d_t|}=0
$$ 
\editt{ Note that it is not apriori clear that the demand would be stable. In fact, we see that in the short-sighted scenario, under certain parameters, the demand is unstable.} 
The definition of the stable market and \eqref{equ:d_t+1} yield the following lemma that is useful in determining the stable outcome of the market. 

\begin{lemma}\label{corollary:quality_24}
The market is stable if and only if the quality $\frac{b_t}{d_t} \xrightarrow[]{t\rightarrow \infty} \frac{1}{\kappa_u}$ is sponsored for end-users.
\end{lemma}
\begin{proof}
The result follows immediately from \eqref{equ:d_t+1}.
\end{proof}
\begin{definition}\label{defintion:stablequality}
\emph{Stable Quality and Stable Demand}: We refer to $\frac{b}{d}=\frac{1}{\kappa_u}$ as the stable quality, and $d=\kappa_u b$ as the stable demand.
\end{definition}

\end{definition}

\section{Sequential Framework: SPNE Analysis}\label{section:analysisI}

In the sequential game framework,  we seek a \emph{Subgame Perfect Nash Equilibrium} (SPNE) using \emph{backward induction}.

\begin{definition}\emph{Subgame Perfect Nash Equilibrium (SPNE):}
	A strategy is an SPNE if and only if it constitutes a Nash Equilibrium (NE) of every subgame of the game. 
\end{definition} 

\begin{definition}\emph{Backward Induction:} Characterizing the equilibrium strategies starting from the last stage of the game and proceeding backward. 
\end{definition}

In this section, we first present the stages of the game (Section~\ref{subsection:stages}). Then, in Section~\ref{section:CPshortSPshort}, we consider the case in which both  the CP and the SP have a short-sighted (myopic) business model and play the one-shot game infinitely. We characterize the equilibrium strategies and asymptotic outcome of the game. When parameters of the market are such that a stable sponsoring outcome is not plausible, considering decision makers with long-sighted vision about the market may ensure a stable sponsoring outcome for the market. Thus, in Sections \ref{section:CPlongSP} and \ref{section:SPlongCP}, using the sequential framework, we investigate the cases in which either one of the SP and the CP is long-sighted and the other is short-sighted. \edit{In Section \ref{section:simulation_seq}, we present numerical results and discuss about them.}

\subsection{Stages of the Game:}\label{subsection:stages}
We suppose a complete information setting for the game. The timing of the game at time epoch $t$ is as follow:

\begin{enumerate}
	\item The SP decides on (1) offering the sponsorship program, $y_t\in\{0,1\}$ (with $y_t=1$ implying offering) and (2)  if $y_t=1$, on  the price per sponsored bit in an LTE frame, $p_t$, by solving the following optimization:
	
	\begin{equation}\label{equ:max_SP}
	\ba 
	&\max_{p_t} u_{SP,t}(p_t),
	\ea
	\end{equation}
	where $u_{SP,t}(p_t)$ is defined in~\eqref{equ:SPutility}. The SP sets $y_t=0$  if $u^*_{SP,t}<v_2 D \log(\kappa_{SP}\frac{N}{D})$ (the payoff is less than no-sponsoring payoff) or $d_t=0$, and $y_t=1$ otherwise, where $u^*_{SP,t}$ is the optimum outcome of the optimization\footnote{Note that we consider that in  the case of indifference $u^*_{SP,t}=0$, $y^*_t=1$}.
	
	\item The CP decides on (1) whether to participate in the sponsorship program,  $z_t\in\{0,1\}$ (with $z_t=1$ implying participation) and (2) if $z_t=1$ on the number of bits in an LTE frame (i.e. quality) she wants to sponsor, $b_t$, by solving the following optimization problem:

	\small
	\begin{equation}\label{equ:optimization_CP_1}
	\ba 
	&\max_{b_t> 0} u_{CP,t}(b_t)\\
	s.t. \quad
	&\qquad\frac{b_t}{d_{t}}\geq \zeta\\
	&\qquad b_t\leq \hat{N}
	\ea 
	\end{equation}
	\normalsize
	where $u_{CP,t}$ is defined in~\eqref{equation:utilityCP}. The first constraint is associated with the minimum quality that the CP wants to deliver to her end-users. The second constraint puts an upperbound to the number of bits that a CP can sponsor in an LTE frame. The CP sets $z_t=0$ if $u^*_{CP,t}<0$ or $d_t=0$, and $z_t=1$ otherwise, where $u^*_{CP,t}$ is the optimum outcome of the optimization\footnote{Note that we consider that in  the case of indifference $u^*_{CP,t}=0$, $z^*_t=1$}. In addition, the CP exits the sponsorship program,  i.e $z_t=0$, if there is no feasible solution for \eqref{equ:optimization_CP_1}. Note that, $
	d_t =
	\l 1+\gamma \log\l\kappa_u \frac{b_{t-1}}{d_{t-1}}\r \r^+ 
	$, and is known as the history of the game is known.

\end{enumerate}

We use the Backward Induction method to find the Sub game Perfect Nash Equilibrium (SPNE) of the game. Thus, first, we find the best response strategy of the CP in the second stage given the strategy of the SP in the first stage and the history of the game. This allows the CP to decide on (1) joining the sponsorship program and also on (2)  the number of bits to sponsor. Then, using this best response strategy and the history, the SP chooses (1) whether to launch the sponsorship program or not, and (2) the optimum per-bit price, $p_t$, in the first stage.

\subsection{Short-Sighted CP, Short-Sighted SP}\label{section:CPshortSPshort}

\textbf{CP's Strategy:}
In the second stage, knowing the decision of the SP at stage one, the CP solves \eqref{equ:optimization_CP_1} at each time-epoch $t$. 

\begin{theorem}\label{theorem:CP_myopic}\textbf{Equilibrium Strategy of Stage 2:}
	The strategy of the CP in the SPNE is as follows:
	\small
	\be \label{equ:opt_CP_1_1}
	\ba 
	\mbox{if } & 0<d_t \leq \frac{\hat{N}}{\zeta},\\
	&(b^*_t,z^*_t) =
	\left\{
	\begin{array}{ll}
		(\hat{N},1)  & \mbox{if }\quad p_t \leq \frac{\alpha d_t}{\hat{N}} \\
		(\frac{\alpha d_t}{p_t},1) &\mbox{if }\quad \frac{\alpha d_t}{\hat{N}}\leq p_t\leq \frac{\alpha}{\zeta}\\
		(\zeta d_t,1)  & \mbox{if }\quad  \frac{\alpha}{\zeta}\leq p_t \leq \frac{\alpha \log(\kappa_{CP} \zeta)}{\zeta} \\
		
		(-,0) &\mbox{if }\quad  p_t > \frac{\alpha \log(\kappa_{CP} \zeta)}{\zeta}  
	\end{array}
	\right.
	\ea
	\ee
	\normalsize
	
	\small
	\be \label{equ:opt_CP_1_2}
	\mbox{if } d_t > \frac{\hat{N}}{\zeta} \text{ or } d_t=0,\quad (b^*_t,z^*_t) =(-,0)
	\ee
	\normalsize
\end{theorem}

\begin{figure}
	\centering
	\includegraphics[width=0.5\textwidth]{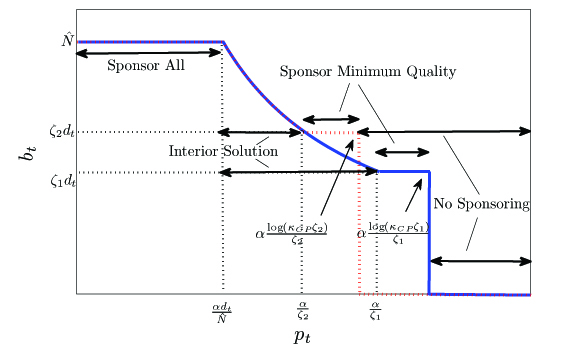}
	\caption{The optimum strategy of the CP presented in theorem~\ref{theorem:CP_myopic} for $\zeta_1$ (blue) and $\zeta_2$ (red) when $0<d_t\leq \frac{\hat{N}}{\zeta}$ and $\zeta_2>\zeta_1$.}\vspace{-4mm}
	\label{figure:Theorem1}
	
\end{figure}

\begin{remark} It is intuitive that the number of sponsored bits is a decreasing function of the price per sponsored bit. In addition, one can expect that if the price per sponsored bit  is lower (respectively, higher) than a threshold, the CP sponsors all the available bits (respectively, the amount to satisfy only the minimum quality requested by the end-users). Moreover, if the price is so high that in case of sponsoring the CP receives a negative payoff, the CP would exit the sponsoring program. This also characterize another threshold for the price per sponsored bit. In theorem~\ref{theorem:CP_myopic}, we confirm the intuitions, and go beyond it by characterizing the thresholds on the price per sponsored bit and optimum number of sponsored bits in different regions characterized by the thresholds. Figure~\ref{figure:Theorem1} illustrates the optimum strategy of the CP and the regions described in the theorem for two different values of $\zeta$. Note that the higher the minimum quality requested by end-users, the lower the thresholds on $p_t$ after which the CP sponsors only the minimum quality or exits the sponsorship program. 
	
	In order to prove the theorem, we apply the first order optimality condition since the utility of the CP is concave. \edit{The proof is presented in the Appendix.}
\end{remark}

\textbf{SP's Strategy:} Now, having the optimum strategy of the CP in stage 2, we can find the optimum strategy for the SP:

\begin{theorem}\label{Theorem:SP}\textbf{Equilibrium Strategy of Stage 1:} 
	The optimum strategies of the SP are:
	\footnotesize
	\be \label{equation:p_eq}
	\ba
	&\text{if }\quad  0<d_t\leq \frac{\hat{N}}{\zeta}, \\
	& (p^*_t,y^*_t) =
	\left\{
	\begin{array}{ll}
		\l \text{argmax}\{u_{SP,t}\l p_t \r: p_t\in P^*\},1\r  & 
		\mbox{if }  u_{SP,t}\l p^*_t\r \geq u_{SP,0} \\
		\l -,0 \r 	&\mbox{if }  u_{SP,t}\l p^*_t\r < u_{SP,0} 
	\end{array}
	\right.\\
	&\text{if }\quad  d_t> \frac{\hat{N}}{\zeta}\text{ or }d_t=0, \quad (p^*_t,y^*_t) =(-,0)
	\ea 
	\ee
	\normalsize
	where $P^*=\{\frac{\alpha d_t}{\hat{N}},\frac{\alpha \log \l \kappa_{CP}\zeta \r}{\zeta},\alpha \frac{\nu_1d_t+\nu_2 D}{\nu_1 N}\}$ is the set of candidate optimum pricing strategies, and $u_{SP,0}$ is considered to be the utility of the SP in case of no-sponsoring, i.e. $v_2 D \log(\kappa_{SP} \frac{N}{D})$. In addition, the necessary condition for the candidate stable point $\alpha \frac{\nu_1d_t+\nu_2 D}{\nu_1 N}$ to be an optimum is $\frac{\alpha d_t}{\hat{N}}\leq \alpha \frac{\nu_1d_t+\nu_2 D}{\nu_1 N}\leq \frac{\alpha}{\zeta}$. Note that the variable $y_t$ determines whether the SP offers the sponsorship program or not, with $y_t=1$ implying the offering.
	
\end{theorem} 
\begin{remark} The immediate plausible range for the price per sponsored bit that one can think of is the interval between the 
	lowest price that makes the CP to sponsor all the available bits and the highest price that makes the CP to sponsor only to satisfy the minimum desired quality. In Theorem~\ref{Theorem:SP}, first, we narrow down this interval to prices between the highest price that makes the CP to sponsor all the available bits and the highest price that makes the CP to sponsor only to satisfy the minimum desired quality. Then, we characterize the interior optimum price. This choice is conditional on getting a payoff greater than or equal to the utility of the SP in case of no-sponsoring. Otherwise, the SP exits the sponsorship program.   
	
	In order to prove the theorem, we use the monotonic behavior of the utility of the SP in some regions, and apply the first order optimality condition for the remaining regions. \edit{The proof is presented in the Appendix.}
\end{remark}

\begin{corollary}\label{remark:indifferent}
	Choosing the price $\frac{\alpha \log(\kappa_{CP} \zeta)}{\zeta}$ by the SP, i.e. the highest price by which the CP sponsors only to guarantee the minimum quality, renders the utility of the CP to be zero, and the CP to be indifferent between joining or not joining the sponsorship program.
\end{corollary}
\begin{proof}
	Results follow from \eqref{equation:utilityCP}, and  that from Theorem~\ref{Theorem:SP}, when $p_t=\frac{\alpha \log(\kappa_{CP} \zeta)}{\zeta}$, $d>0$ and $b_t=\zeta d_t$ (from Theorem~\ref{theorem:CP_myopic}).
\end{proof}

\textbf{Outcome of the Game:} Now that we have characterized the SPNE at each time-epoch for a short-sighted CP and SP, the next step is to analyze the asymptotic behaviour of the market given the demand update function \eqref{equ:d_t+1} and considering the one-shot game to be repeated infinitely. The goal is to characterize the asymptotically stable 5-tuple equilibrium outcome of the game, i.e. $\l d,y,p,z,b \r$ (table~\ref{table:symbols}), if it were to exist. In the next Theorem, all possible asymptotically stable outcomes are listed. However, the existence of such a stable outcome is not guaranteed, and the market can be unstable in some cases. 


\begin{theorem}\label{theorem:stable}
	The possible asymptotically stable  outcomes of the game are:
	\begin{enumerate}
		\item $\l -,0,-,0,-\r$: \emph{no sponsoring} is offered, none taken.
		\item $\l \kappa_u \hat{N}, 1, \alpha \kappa_u, 1, \hat{N}\r$: \emph{the maximum bit sponsoring};  if this is the stable outcome then $\kappa_u\leq \frac{1}{\zeta}$.
		\item $\l d,1,\frac{\alpha \log\l \kappa_{CP}\zeta\r}{\zeta},1,\zeta d\r$: \emph{the minimum quality sponsoring}; if this is the stable outcome then $\kappa_u=\frac{1}{\zeta}$ and $0<d\leq \frac{\hat{N}}{\zeta}$.
		
		\item  $\l N \kappa_u - \frac{\nu_2}{\nu_1}D, 1, \alpha \kappa_u,1, N - \frac{\nu_2}{\kappa_u \nu_1}D\r$: \emph{the interior stable points}; if this is the stable outcome, then $\kappa_u\leq \frac{1}{\zeta}$ and $0<b=N - \frac{\nu_2}{\kappa_u \nu_1}D\leq \hat{N}$.
	\end{enumerate}
\end{theorem}

\begin{remark}
	Since the CP is shortsighted, in every stable outcome of the game, the strategy of the CP would be a myopic optimum strategy. Thus, using Theorem~\ref{theorem:CP_myopic}, one can expect the strategy of the CP to take one of the four possibilities in a stable outcome: (1) no sponsoring (2) sponsoring the maximum amount available (3) sponsoring only to satisfy the minimum required quality, or (4) sponsoring an optimum interior amount of bits. Subsequently, depending on the strategy of the CP, Theorem~\ref{Theorem:SP} characterizes the stable strategy of the SP. In order to prove the theorem we use Lemma~\ref{corollary:quality_24}. \edit{The proof is presented in the Appendix.}
\end{remark}


\begin{corollary}\label{corollary:overprovision}
	There is no stable outcome involving sponsoring for the game if the stable quality is smaller than the minimum quality set by the CP, i.e. $\frac{1}{\kappa_u}<\zeta$. 
\end{corollary}

\begin{remark}
	Unlike other plausible stable outcomes, the third possible stable point, i.e. $\l d,1,\frac{\alpha \log\l \kappa_{CP}\zeta\r}{\zeta},1,\zeta d\r$ when  $\zeta =\frac{1}{\kappa_u}$, can assume a range of different values. Whenever the SP sets $p=\frac{\alpha \log\l \kappa_{CP}\zeta\r}{\zeta}$, the CP sets $\frac{b_t}{d_t}=\zeta$, and the market will be stable. By choosing that price, the SP ensures that she will extract all the profit of the CP and makes her indifferent between joining the sponsorship program and opting out, i.e. $u_{CP}(b)=0$ (using Corollary~\ref{remark:indifferent}). 
\end{remark}


In the next theorem, we find the stable demand that maximizes the payoff of the SP when she chooses the third stable point, i.e. the minimum quality.  

\begin{theorem}\label{theorem:SPd}
	Let $d^*=\frac{N}{\zeta}-\frac{1}{\l \alpha+\nu_1\r \log\l \kappa_{SP}\zeta\r}$. The payoff of the SP when the 5-tuple  $\l d,1,\frac{\alpha \log\l \kappa_{CP}\zeta\r}{\zeta},1,\zeta d\r$ (the minimum quality stable point) is the stable outcome of the market is  maximized when either (1) $d=\min\{d^*,\frac{\hat{N}}{\zeta}\}$ and $d^*\geq 0$, or (2) $d=0$ and $d^*<0$.
\end{theorem}

The proof of the theorem is presented in the appendix.



In the next sections,  we investigate the case in the SP is long-sighted and the CP is short-sighted.

\subsection{Long-Sighted SP, Short-Sighted CP} \label{section:CPlongSP}

A long-sighted SP sets the per-bit sponsorship fee in order to achieve a stable market, i.e. a stable demand for the content, and also to maximize the payoff in the long-run:
\begin{equation}\label{equ:long_SP_max}
U_{SP,Long\ Run}(\vec{p})= \lim_{T\rightarrow \infty}\frac{1}{T}  \sum_{t=1}^T u_{SP,t}(p_t)
\end{equation}
\normalsize

In this scenario, the SP is the leader of the game and therefore can set the equilibrium of the game individually by knowing that the CP is a myopic optimizer unit and follows the results in Theorem~\ref{theorem:CP_myopic}. Note that the long-sighted SP wants to asymptotically set a strategy that given the strategy of the CP, yields the highest profit. Thus, even with a long-sighted SP, the optimum strategy follows Theorem~\ref{Theorem:SP}, and we can use the result in Theorem \ref{theorem:stable}. 


The difference between this case and the previous case is the ability of the long-sighted SP to choose between the candidate stable points in  Theorem~\ref{theorem:stable}. Thus, the SP sets appropriate sponsoring fees at the beginning of the sponsoring program in order to asymptotically lock the stable outcome of the market in the chosen equilibrium.  

Note that from Theorem~\ref{theorem:stable}, when $\kappa_u>\frac{1}{\zeta}$, there is no stable sponsoring outcome, and if $\kappa_u<\frac{1}{\zeta}$, depending on the parameters of the market, the stable point 2, i.e. maximum bit sponsoring, or 4, i.e. interior stable point, is chosen by the SP. In this case, if $\nu_2$, i.e. the importance of non-sponsored data for end-users and SP, is high enough, the stable point 4 is chosen and set by the SP. In addition, increasing the  number of resources available with the SP, i.e. $N$, makes the stable point 2, i.e. maximum bit sponsoring, to yield the highest payoff, and thus is chosen by the SP. In the next theorem, we prove that when $\kappa_u=\frac{1}{\zeta}$, the  stable point 3 yields the highest payoff. 

\begin{theorem}\label{Theorem:longrunSP}
	If $\kappa_u=\frac{1}{\zeta}$, the minimum quality stable point, i.e. $\l d,1,\frac{\alpha \log\l \kappa_{CP}\zeta\r}{\zeta},1,\zeta d\r$, with the demand characterized in Theorem~\ref{theorem:SPd} yields the highest payoff for the SP. 
\end{theorem}

\begin{remark}
	Note that in a minimum quality stable point, the CP is indifferent, i.e. all profit of the CP is extracted by the SP. Therefore, we can expect this stable outcome to be the most favorable for the SP. Thus, a long-sighted SP sets this stable point as the asymptotic outcome of the market when $\kappa_u=\frac{1}{\zeta}$. 
\end{remark}
The proof of the theorem is presented in the appendix.

\subsection{Short-Sighted SP, Long-Sighted CP}\label{section:SPlongCP}

Consider a CP that chooses $b_t$ in order to achieve a stable demand  and to maximize the payoff in the long-run:
\small
\begin{equation}\label{equ:long_CP_max}
U_{CP,Long\ Run}(\vec{b})= \lim_{T\rightarrow \infty}\frac{1}{T}  \sum_{t=1}^T u_{CP,t}(b_t)
\end{equation}
\normalsize

In the next theorem, we prove that for a long-sighted CP, the maximum bit sponsorship yields the highest payoff amongst the stable outcomes characterized in Theorem~\ref{theorem:stable}. Note that if the CP sponsors all the available units at the start of the sponsoring program, the sudden increase in the demand may push the market to the stable outcome of no sponsoring. Thus, given that the SP is short-sighted, the CP sets the number of bits for sponsoring appropriately over time, in order to achieve the demand of $\kappa_u \hat{N}$ eventually. \edit{With this demand and  $b=\hat{N}$, the market would be stable. However, note that not the SP is the leader of the game and may chooses a price other than $\alpha \kappa_u$, i.e. the price in the maximum bit sponsoring. In this case, the CP cannot set the stable outcome she prefers. Thus, the CP is forced to set a stable outcome that is also preferable for a short-sighted SP.}

\begin{theorem}\label{theorem:CP_long_sighted}
	The 5-tuple plausible stable sponsoring points, characterized in Theorem~\ref{theorem:stable}, in a decreasing order of the utility they yield to the CP are:  maximum bit sponsorship, interior stable point, and minimum quality.
\end{theorem}

\begin{remark}
	In order to establish the results, note that the stable point of no sponsoring is not considered a \emph{sponsoring} stable point. Thus not listed in the theorem. 
	In addition, since the CP is indifferent in the minimum quality stable point, this point should be the least favorite one for the CP. The ordering of the maximum bit and the interior stable points follows from the fact that the payoff of the CP is strictly increasing in those outcomes.  \edit{The proof is presented in the Appendix.}
\end{remark}
\begin{figure}[t]
	\centering
	\includegraphics[width=0.4\textwidth]{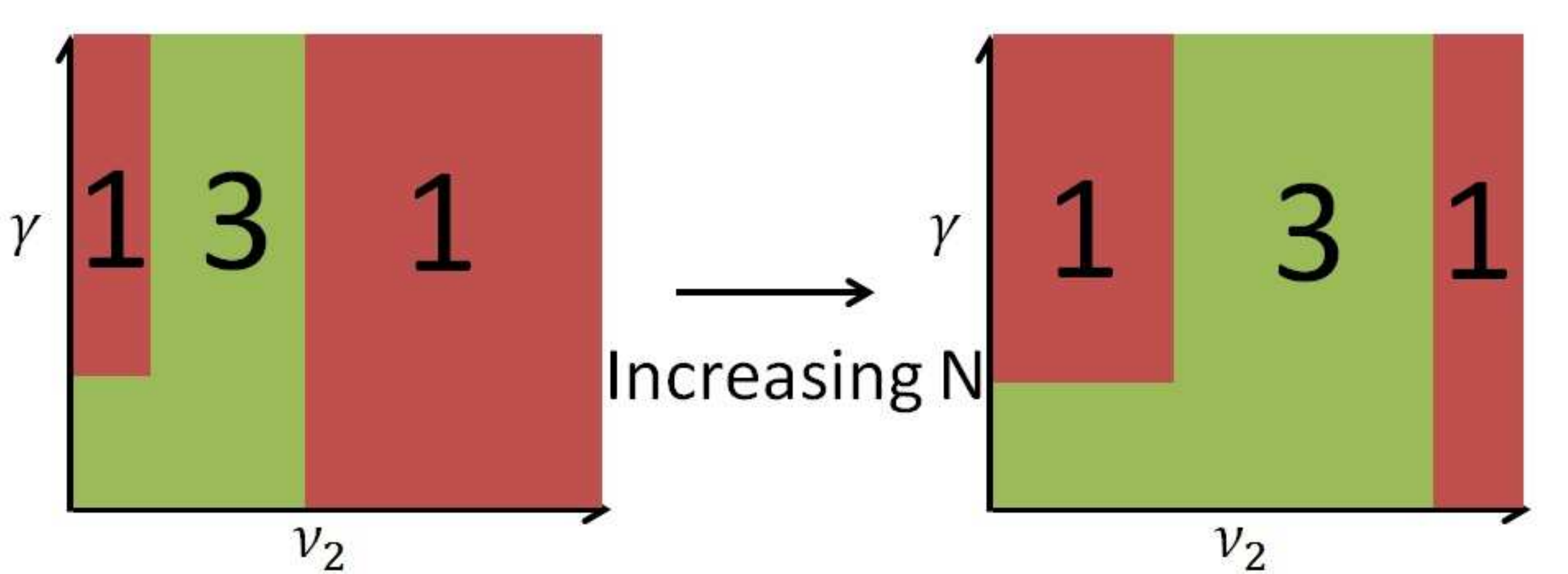}
	\caption{Market Asymptotic Outcomes with Short-Sighted Decision Makers when $\kappa_u =\frac{1}{\zeta}$}
	\label{figure:shortsighted_equal}
\end{figure}
\begin{figure}
	\centering
	\includegraphics[width=0.4\textwidth]{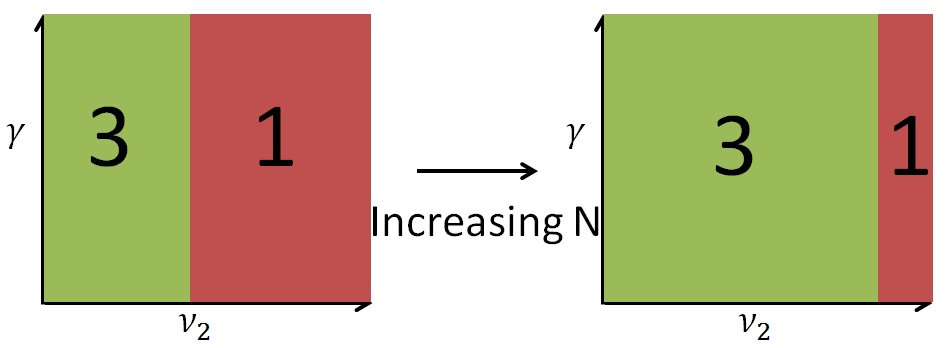}
	\caption{Market Asymptotic Outcomes when One of the Decision Makers is Long-Sighted, when $\kappa_u =\frac{1}{\zeta}$}
	\label{figure:shortsighted_longsighted_equal}
\end{figure}

\subsection{Numerical Results}\label{section:simulation_seq}

In this section, we consider \edit{at least one of} the SP and the CP to have a short-sighted  model, and investigate the effects of $\zeta$, $\kappa_u$, $v_2$, $N$, and $\gamma$ on the asymptotic outcome of the market.  The fixed parameters considered  are $\nu_1=1$, $\hat{N}=25$, $D=50$, $\kappa_{SP}=\kappa_{CP}=10$, and $\zeta=0.3$. We  observe the effect of important parameters such as  $\kappa_u$, the sensitivity of the demand to the quality ($\gamma$), the weight an SP assigns to non-sponsored data ($\nu_2$), and the total number of available bits in an LTE frame ($N$) on the asymptotic outcome of the market. 

\edit{Market asymptotic outcomes for the case of $\kappa_u=\frac{1}{\zeta}$ when both decision makers are short-sighted and when one of them is long-sighted are presented in Figures~\ref{figure:shortsighted_equal} and \ref{figure:shortsighted_longsighted_equal}, respectively. Similar plots for the case of   $\kappa_u=\frac{1}{2\zeta}<\frac{1}{\zeta}$ are presented in Figures ~\ref{figure:shortsighted_lessthan} and \ref{figure:shortsighted_longsighted_lessthan}, respectively.} Recall from Theorem~\ref{theorem:stable} that the asymptotic stable outcome of the game is one the four candidates: 1. No-Sponsoring, 2. Maximum bit sponsoring, 3. Minimum quality, and 4. Interior stable point. In the figures, each of the numbers are corresponding to one of the  candidates. We  also denote the unstable outcome by $0$.

Next, we discuss about the effect of \edit{the framework} and parameters on the asymptotic outcome of the game:

{\bf Impact of a decision-maker with long-sighted model:} \edit{Note that in both cases $\kappa_u=\frac{1}{\zeta}$ and $\kappa_u<\frac{1}{\zeta}$,  the outcome of the market is independent of  the CP or the CP being long-sighted or the SP being long-sighted. The reason is that in the sequential game the SP is the leader of the game. Thus, although a long-sighted SP can set the stable outcome she prefers in the long-run, a long-sighted CP cannot enforce the most preferred stable outcome, and should choose the stable outcome that is also preferable for the  SP.  This yields the same asymptotic outcome for the market as the case that the SP is long-sighted.}

{\bf Impact of the minimum quality ($\zeta$) and the stable quality ($\frac{1}{\kappa_u}$):} Theorem \ref{theorem:stable} implies that depending on the relation between the minimum quality set by the CP ($\zeta$) and the stable quality ($\frac{1}{\kappa_u}$), the market has different stability outcomes in the short-sighted scenario. We seek to identify the stable outcomes that arise in each different parameter ranges:

$\bullet$ Using Corollary~\ref{corollary:overprovision}, If the CP \emph{over-provisions} the minimum quality for the satisfaction of users ($\zeta>\frac{1}{\kappa_u}$), there is no stable sponsoring outcome since the demand of users grows drastically forcing the SP and CP to exit the sponsoring program.  Thus, we do not study this case through simulations.

\begin{figure}[t]
	\centering
	\includegraphics[width=0.4\textwidth]{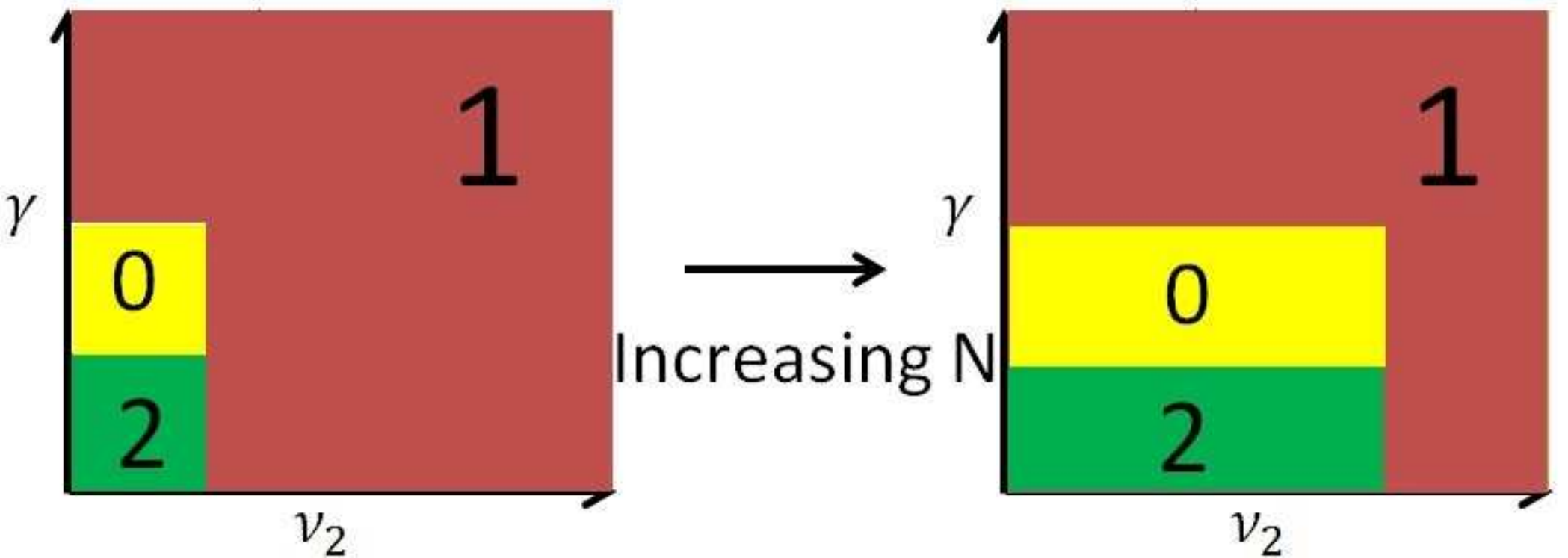}
	\caption{Market Asymptotic Outcomes with Short-Sighted Decision Makers, when $\kappa_u < \frac{1}{\zeta}$}
	\label{figure:shortsighted_lessthan}
\end{figure}

\begin{figure}
	\centering
	\includegraphics[width=0.4\textwidth]{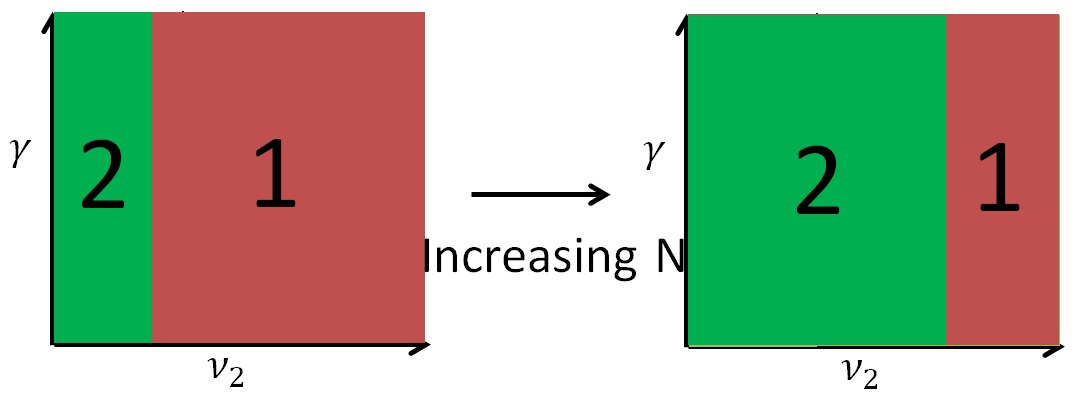}
	\caption{Market Asymptotic Outcomes when One of the Decision Makers is Long-Sighted, when $\kappa_u < \frac{1}{\zeta}$}
	\label{figure:shortsighted_longsighted_lessthan}
\end{figure}


$\bullet$ If the CP \emph{under-provisions} the minimum quality ($\zeta<\frac{1}{\kappa_u}$), simulation results in Figure~\ref{figure:shortsighted_lessthan} reveal that the market is set on the maximum bit sponsoring stable outcome for a particular range of parameters. However, the market is unstable or has the stable outcome of no-sponsoring for the rest of parameters.

$\bullet$ If the CP sets the minimum quality equal to the stable quality ($\zeta=\frac{1}{\kappa_u}$), the market is set either on the no-sponsoring or the minimum quality stable outcomes. Comparing Figures~\ref{figure:shortsighted_equal}, \ref{figure:shortsighted_longsighted_equal}, \ref{figure:shortsighted_lessthan}, and \ref{figure:shortsighted_longsighted_lessthan} reveals that in the case of $\kappa_u=\frac{1}{\zeta}$, the market is more likely to have a stable outcome that involves sponsoring.





{\bf Impact of the sensitivity of the demand to the quality, $\gamma$:}
Note that, in Figures \ref{figure:shortsighted_equal} and \ref{figure:shortsighted_lessthan}, and in general, increasing the value of $\gamma$ shifts the stable outcome of the market from sponsoring to no-sponsoring. The exception is a range of $v_2$ for the case $\kappa_u=\frac{1}{\zeta}$, which we explain about when we  discuss about the impact of $v_2$ later.


Therefore, in a market with short-sighted entities, the sensitivity of the demand to the quality, i.e. $\gamma$, greatly influences the stability of the market. When $\gamma$ is high, the satisfaction and subsequently the demand of end-users increases/decreases drastically with small changes in the rate perceived by them. Thus, players would exit the sponsorship program  since the demand may go down to zero or the demand may exceed $d_{max}=\frac{\hat{N}}{\zeta}$\footnote{the highest number of end-users that can be satisfied with the minimum quality.}, \edit{i.e. the jump in the demand decreases the quality received by the users below the requested minimum quality ($\zeta$) which leads the CP to stick to the best-effort scenario.}
On the other hand, if $\gamma$ is small, the market is more likely to be set on a sponsoring stable outcome.


Thus, in order to have a stable outcome of sponsoring, $\gamma$ should be sufficiently small. Note that this parameter is small for a CP whose users are less sensitive to the quality they receive, such as shopping websites. This is in contrast with streaming websites whose users are sensitive to the quality (high $\gamma$). The parameter $\gamma$ is also small for a CP which has a well-established end-user side, i.e. a more stable demand, such as Google, in contrast with the emerging CPs and start ups whose demand usually fluctuate more. Thus, in a short-sighted setting, the QSD may not be a viable option for streaming websites and emerging CPs.  

\edit{In addition, note that the effect of $\gamma$ would be canceled if one of the decision makers is long-sighted. This implies that in QSD framework, CPs with volatile demand (high $\gamma$) should be long-sighted to have a market with a stable outcome that involves sponsoring the content.}

{\bf Impact of the importance of non-sponsored contents, $\nu_2$:} The parameter $v_2$ being large, when $v_1$ is normalized to one, represents the fact that the SP assigns more weight to the satisfaction of users for using non-sponsored content. Thus, the SP wants to restrict the number of bits she offers for sponsoring, and the best strategy for the SP is to set her per-bit sponsorship fee high enough so that the CP sponsors a smaller number of bits. Thus, we expect the market to have a stable outcome of no-sponsoring when $v_2$ is high. 

Results in Figures \edit{ \ref{figure:shortsighted_equal}, \ref{figure:shortsighted_longsighted_equal},  \ref{figure:shortsighted_lessthan}, and \ref{figure:shortsighted_longsighted_lessthan} } confirm that the market has the stable outcome of no-sponsoring when $v_2$ is large. One of the differences between the cases $\kappa_u=\frac{1}{\zeta}$ and $\kappa_u<\frac{1}{\zeta}$ is that when $\kappa_u=\frac{1}{\zeta}$, for a certain range of $v_2$, the stable sponsoring outcome is 3, i.e. the minimum quality stable outcome, regardless of $\gamma$, i.e. the sensitivity of the demand to the quality. Next, we explain the reason for this behavior. Note that, as we mentioned, $v_2$ being high is associated with lower bits sponsored. Thus, for a certain range of $v_2$, we expect the CP to start the sponsoring program with a quality near the minimum quality (since the CP wants to sponsor at least the minimum quality). In the case that $\kappa_u=\frac{1}{\zeta}$, \eqref{equ:d_t+1} implies that the demand increases more slowly (the logarithm in the expression is smaller). Thus, the effect of $\gamma$ is not significant, and the market can be stabilized on the minimum quality stable point regardless of $\gamma$. However, this does not happen for the case of $\kappa_u<\frac{1}{\zeta}$, since in this case, from \eqref{equ:d_t+1}, the demand of end-users diminishes to zero\footnote{Note that in this case the logarithm in  \eqref{equ:d_t+1} is negative for a quality near the minimum quality.}, and the market is set on the stable outcome of no-sponsoring.



{\bf Impact of total available resources, $N$:} \edit{Figures~\ref{figure:shortsighted_equal}, \ref{figure:shortsighted_longsighted_equal},  \ref{figure:shortsighted_lessthan}, and \ref{figure:shortsighted_longsighted_lessthan}  reveal that increasing the number of available resources ($N$) stretches the regions. In other words, as $N$ increases, results would be similar to that of smaller $v_2$'s. }

For example, in Figure~\ref{figure:shortsighted_equal}, increasing the number of available bits (resources) increases the area of no-sponsoring region \edit{for small $v_2$}. This is \emph{counter-intuitive} since one can expect that increasing the amount of available resources should facilitate sponsoring the content. This counter-intuitive result is due to the fact that by increasing $N$, the value of the SP for each bit decreases and the SP sets a lower sponsoring fee. This leads to sponsoring more bits by the CP which leads to a significant increase in the demand for the content when $\gamma$ is large. This derives the market to the point of no-sponsoring. Therefore, the outcome is the same as the case in which $v_2$ is very small: the minimum quality stable point when $\gamma$ is small and no-sponsoring when $\gamma$ is large. 

{\bf Impact of $v_1$ and $v_2$ on the social welfare:}
\edit{If we define the social welfare of the QSD regime as the sum of the payoffs of the CP and the SP\footnote{\edit{Note that the payoff of the SP includes a term for users' satisfaction function that captures the welfare of EUs for sponsored and non-sponsored contents (possibly with constants different from $v_1$ and $v_2$). In addition, the effect of the model on other CPs is also hidden in the users' satisfaction function (the term $v_2 D \log(\kappa_{SP}\frac{N-b_t}{D})$). Thus, sum of the utility of the CP and the SP (with possibly different $v_1$ and $v_2$) is a good indicator of the social welfare.}}, then important parameters for determining the social welfare of the system are $v_1$ and $v_2$ (can be imposed by the regulator), and $D$. In this case, if  $\frac{v_2}{v_1}$ or $D$ are high, i.e. when the weight on the content of a non-sponsored content is high,  then the SP restricts the number of bits she offers for sponsoring by quoting a high sponsorship fee (as explained before). Thus, either the CP reserves a smaller number of bits or exits the sponsorship program. In either cases, the outcome would be aligned with maximizing the social welfare. Thus, $v_1$ and $v_2$ can be imposed by the regulator to control the social welfare.}




\begin{remark} Figure~\ref{figure:shortsighted_lessthan} and \ref{figure:shortsighted_longsighted_lessthan} illustrate that the stable point $4$, i.e. interior stable point, does not emerge, and only stable points $1$ and $2$ occur. In other words, in a stable outcome, when stakeholders of the market are short-sighted, either the CP sponsors all the available resources or no sponsoring occurs. Note that in the stable point $4$, the number of bits sponsored by the CP in the equilibrium is $N-\frac{\nu_2}{\kappa_u \nu_1}D$. In addition, a stable sponsoring 5-tuple occurs only when $\nu_2$ is small which makes $N-\frac{\nu_2}{\kappa_u \nu_1}D>\hat{N}$ for a wide range of parameters. Thus, the stable point $4$ does not emerge in many scenarios. One can argue that by decreasing $N$ or increasing $D$, we may have a scenario in which $N-\frac{\nu_2}{\kappa_u \nu_1}D<\hat{N}$. However, note these changes,  is similar to having a large $v_2$. Thus, similar to previous arguments, in this case, the SP is willing to set a price so high that leads the CP and market to a no-sponsoring outcome. Therefore, again in the regions that support an interior sponsoring solution, i.e. when $N-\frac{\nu_2}{\kappa_u \nu_1}D<\hat{N}$, the stable outcome $4$ would not occur. None of the parameters we considered results in such a stable outcome. 
\end{remark}

In the next section, using a bargaining framework, we  investigate the scenario in which the decision makers have a \emph{long-sighted  model}, i.e. consider the effect of their decisions on the evolution of the demand and subsequently their payoff.


\section{Bargaining Framework: NBS Analysis} \label{section:bargaining}

In the previous section, we proved that a long-sighted CP and SP can prefer different stable outcomes, i.e. the stable outcomes that yield the highest payoff for them. If both decision makers are long-sighted, since  \edit{multiple asymptotic outcomes are plausible}, playing a sequential game may lead to a Pareto-inefficient outcome \edit{in the long-run}\footnote{\edit{
		The strategies of the SP and the CP are Pareto-inefficient in the long-run if	at least one of  the CP or the SP  can increase her payoff,  by changing her strategy, without decreasing the other player's payoff. An example of an inefficient outcome that occurs in our model is when $v_2$ is small and $\gamma$ is large and both players are short-sighted. According to Figure~\ref{figure:shortsighted_equal} the asymptotic outcome of the game would be the no-sponsoring outcome. On the other hand, in Figure \ref{figure:shortsighted_longsighted_equal}, with the same parameters, when one of the decision makers is long-sighted (which means that she chooses different strategies to maximize her long-run payoff), the asymptotic outcome of the market would be the minimum quality sponsoring outcome in which the SP receives a strictly higher payoff than the previous case. Thus, the  outcome of the sequential game, can be Pareto inefficient in the long-run.}}. 
Therefore, when both of the decision makers are long-sighted, it is natural to consider a bargaining game framework\footnote{We can also consider a bargaining game when decision makers are short-sighted. However, in this paper, we consider two extreme scenarios: (1) non-cooperation/at least one decision maker short-sighted and (2) cooperation/both long-sight-sighted, and compare the outcome of the market in these two extreme cases.}. A bargaining game provides the framework to model the scenario in which two selfish agents can cooperatively  select an equilibrium outcome (possibly among multiple equilibria)  when non-cooperation, i.e. \emph{disagreement}, yields Pareto-inefficient results Note that both cases (multiple Nash equilibria and Pareto-inefficient outcome) occur in our modeling in Section~\ref{section:analysisI} when at least one of the decision makers is short-sighted. After selecting the equilibrium, the division of profits can be characterized using the bargaining game frameworks considering the \emph{bargaining power} of each decision maker.

\edit{In Section~\ref{section:bargaining_analy}, we formulate and analyze the bargaining game. In Section~\ref{section:bargaining_simulation}, we present numerical results for this framework and discuss about the results.}



\subsection{Nash Bargaining Solution (NBS)}\label{section:bargaining_analy}

Thus, we formulate the interaction between the CP and the SP as a bargaining game, and  use the \emph{Nash Bargaining Solution} (NBS) to characterize the bargaining solution to the problem when both the SP and the CP are long-sighted. 


\begin{definition}\emph{Nash Bargaining Solution (NBS):} is the unique solution ( in our case the tuple of the payoffs of the CP and the SP) that satisfies the four ``reasonable" axioms (Invariant to affine transformations, Pareto optimality, Independence of irrelevant alternatives, and Symmetry) characterized in \cite{osbornebargaining}.
\end{definition}

Let $0\leq w\leq 1$ be the relative bargaining power of the CP over SP: the higher $w$, more powerful is the bargaining power of a CP. In addition, $u_{CP}$ and $u_{SP}$ denote the payoff of the CP and SP respectively, and $d_{CP}$ and $d_{SP}$ denote the payoff each decision maker receives in case of disagreement, i.e. \emph{disagreement payoff}. In order to characterize the disagreement payoffs, we assume that in case of disagreement, the SP and the CP will interact as short-sighted entities playing the sequential game previously described\footnote{\edit{The reason is that if in the case of disagreement, the CP and the SP continue their selfish non-cooperative behavior, they can obtain  a payoff higher than or equal to the payoff of no-sponsoring. The inequality is strict for the cases that a sequential game yields a sponsoring outcome.}}. Thus, the disagreement payoffs can be found by determining the asymptotic status of the market: the asymptotic payoff of the CP and the SP if the market is asymptotically stable, or the average payoffs if the market is unstable. Note that the value of the disagreement payoff for the CP and the SP can have an effect similar to the bargaining power ($w$ for the CP and $1-w$ for the SP). 

Using standard game theoretic results in \cite{osbornebargaining}, the pair of $u^*_{CP}$ and $u^*_{SP}$ can be identified as the Nash bargaining solution of the problem if and only if it solves the following optimization problem:
\begin{equation}\label{equ:nash_solution_1}
\begin{aligned}
&\max_{u_{CP}, u_{SP}}{(u_{CP}-d_{CP})^{w}(u_{SP}-d_{SP})^{1-w}}\\
& \text{s.t.}\\
&\qquad (u_{CP},u_{SP})\in U\\
& \qquad (u_{CP},u_{SP})\geq (d_{CP},d_{SP})
\end{aligned}
\end{equation}
where $U$ is the set of feasible payoff pairs. Note that the long-sighted SP and CP want to set a stable market in the long-run\footnote{\edit{An unstable market makes it difficult for the CP and the SP to make decisions, predict the demand, or manage the network. Thus, an unstable market has its implicit costs for the CP and the SP. This is the reason that we assumed that the CP and the SP want to set an stable market.}}, and based on Lemma~\ref{corollary:quality_24} in a stable outcome $\frac{b}{d}=\frac{1}{\kappa_u}$. Thus, the expressions for $u_{CP}$ and $u_{SP}$ in a stable outcome (using \eqref{equation:utilityCP} and \eqref{equ:SPutility}) are functions of the demand ($d$) and as follows: 
\begin{equation}\label{equ:uCPbarg}
u_{CP}(d)=u_{ad}(d)-p \frac{d}{\kappa_u}
\end{equation}
\begin{equation}\label{equ:uSPbarg}
u_{SP}(d)= p \frac{d}{\kappa_u}+u_s(d)
\end{equation}
where $u_{ad}(d)=\alpha d \log(\frac{\kappa_{CP}}{\kappa_u})$ is the advertisement profit for the CP, and $u_s(d)=\nu_1 d \log(\frac{\kappa_{CP}}{\kappa_u})+\nu_2 D\log(\kappa_{SP}\frac{N-\frac{d}{\kappa_u}}{D})$ is the satisfaction of the end-users of the SP. In addition, $p\frac{d}{\kappa_u}$ is the side-payment transferred from the CP to the SP in exchange of securing a quality of $\frac{1}{\kappa_u}$ for the demand ($d$). Note that in Section~\ref{section:model}, we introduced $u_{ad}(.)$ and $u_{s}(.)$ as functions of the number of sponsored bits ($b$). Here, we redefine them to be functions of demand ($d$) since in a stable outcome $d=\kappa_u b$. 

Thus, the maximization \eqref{equ:nash_solution_1} is over $d>0$ and $p$. In addition, note that the maximum demand that can be satisfied with maximum available resources of $\hat{N}$ to provide the quality of $\frac{1}{\kappa_u}$ is $\kappa_u\hat{N}$ ($d=\kappa_u b\leq \kappa_u \hat{N}$), which constitutes the feasible set. Thus the maximization is,
\begin{equation}\label{equ:nash_solution_2}
\begin{aligned}
&\max_{d, p}{(u_{CP}-d_{CP})^{w}(u_{SP}-d_{SP})^{1-w}}\\
& \text{s.t.}\\
&\qquad 0\leq d\leq \hat{N}\kappa_u \\
&\qquad u_{CP}\geq d_{CP}\\
&\qquad u_{SP}\geq d_{SP}
\end{aligned}
\end{equation}

We define $p^*$ and $d^*$ to be the optimum solution of \eqref{equ:nash_solution_2}. Note that $p^*$ and $d^*$ characterize the optimum  division of profit ($u^*_{CP}$ and $u^*_{SP}$) and thus the NBS. In addition, we define the \emph{aggregate excess profit} to be the additional profit yielded from the cooperation in the bargaining framework:
\begin{definition}\emph{Aggregate Excess Profit ($u_{excess}$):} The aggregate excess profit is defined as follows:
	\vspace{-1mm}
	\begin{equation}
	u_{excess}=u_{CP}-d_{CP}+u_{SP}-d_{SP}=u_{ad}-d_{CP}+u_s-d_{SP}
	\end{equation}
	
\end{definition}
Note that $u_{excess}$ in independent of $p$ and is only a function of $d$. We define $u^*_{excess}= u_{\text{excess}}|_{d=d^*}$. Note that the bargaining would only occur if $u^*_{excess}>0$, i.e.  the framework creates additional joint profit that can be divided between the SP and the CP.  Thus, henceforth, we characterize the NBS for the case that $u^*_{excess}>0$. We use $u_{excess}$ in the following theorem:

\begin{theorem}\label{theorem:bargaining}
	If  $u^*_{\text{excess}}> 0$, the optimum solution of the optimization \eqref{equ:nash_solution_2} is $(d^*,p^*)$ in which $d^*=\text{arg}\max_{0\leq d\leq \hat{N}\kappa_u}{u_{\text{excess}}}$, and $p^*$ is:
	\begin{equation}\label{equ:optimum_p}
	\begin{aligned}
	p^*&=\frac{\kappa_u}{d} \Big{(}(1-w)(u^*_{ad}-d_{CP})-w(u^*_s-d_{SP})\Big{)}\\
	&=\frac{\kappa_u}{d} \Big{(}(u^*_{ad}-d_{CP})-wu^*_{excess}\Big{)}
	\end{aligned}
	\end{equation}
	where $u^*_{ad}=u_{ad}|_{d=d^*}$, $u^*_s=u_s|_{d=d^*}$.

\end{theorem}

\begin{remark}
	The theorem characterizes $p^*$ and $d^*$ which directly lead to the NBS (using \eqref{equ:uCPbarg} and \eqref{equ:uSPbarg}), i.e. $(u^*_{CP},u^*_{SP})$. Based on the  theorem, before splitting the profit, the SP and the CP cooperatively set a stable quality and subsequently a stable demand to maximize the aggregate excess profit, $u_{excess}$ by solving the concave maximization problem ($\max_d u_{excess}$) with the single parameter $d$. 
	Subsequently, they decide the split of the additional profit, i.e. the side payment paid to the SP by the CP ($p^* \frac{d^*}{\kappa_u}$), based on \eqref{equ:optimum_p} which depends on the bargaining power each has ($w$ and $1-w$). \edit{The proof of Theorem is presented in the Appendix.}
\end{remark}

\begin{remark}
	As we mentioned before, the value of the disagreement payoffs can also play a similar role as the bargaining power ($w$). From \eqref{equ:optimum_p}: $d_{CP}\uparrow \Rightarrow p\downarrow$, and $d_{SP}\uparrow \Rightarrow p\uparrow$.    
\end{remark}

\begin{remark}{\textbf{Price vs. Bargaining Power:}} The price per sponsored bit \eqref{equ:optimum_p} is a decreasing function of $w$ , i.e. the bargaining power of the CP: the higher the bargaining power of the CP the lower the side payment paid to the SP. It follows from \eqref{equ:optimum_p} that there exists a threshold on $w$, $w_t=\frac{u^*_{ad}-d_{CP}}{u^*_{excess}}$, such that  when $w>w_t$, $p^*<0$, when $w<w_t$, $p^*>0$, and when $w=w_t$, $p^*=0$. In other words, for the CP with a bargaining power higher than the threshold $w_t$, the flow of money is reversed, and the SP pays the CP. This counter-intuitive case occurs  either due to a high bargaining power of the CP (high $w$), or in the scenario that the SP gain significantly more than the CP from the cooperative scenario ($u^*_{excess}>>u^*_{ad}-d_{CP}$, i.e. low $w_t$). For example, a powerful CP, e.g. Google, which already has a large established  demand for the content might be reluctant to cooperate with the SP unless the SP pays some of the additional profit to it.   
\end{remark}

\subsection{Numerical Results}\label{section:bargaining_simulation}
Now, consider the SP and the CP with long-sighted business model 	that play a bargaining game as described in this Section.  We investigate the effects of bargaining and cooperation between the CP and the SP in increasing the utility of each of them and stabilizing the market. In addition, we discuss about the relation between the number of available resources ($N$) and the Nash bargaining price ($p^*$).

We consider $w=0.5$, i.e. the CP and the SP have the same bargaining power, and the values of $v_1$, $\hat{N}$, $D$, $\kappa_{SP}$, and $\kappa_{CP}$ to be the same as  those considered in Section~\ref{section:simulation_seq}.

First, we plot the percentage of increase in the payoff of the CP and the SP after bargaining versus $v_2$ for different values of $\zeta$ when $\kappa_u=\frac{1}{\zeta}$ in Figures~\ref{figure:increase_CP_equal} and \ref{figure:increase_SP_equal}.  The percentage of increase in the utility after bargaining is defined as follows:

\begin{equation}\label{equ:increase}
\begin{aligned}
&\text{increase (percentage)}=\\
&\quad \frac{\text{utility after bargaining-utility before bargaining}}{\text{utility after bargaining}}\times 100
\end{aligned}
\end{equation}

\begin{figure}[t]
	\centering
	\includegraphics[width=0.4\textwidth]{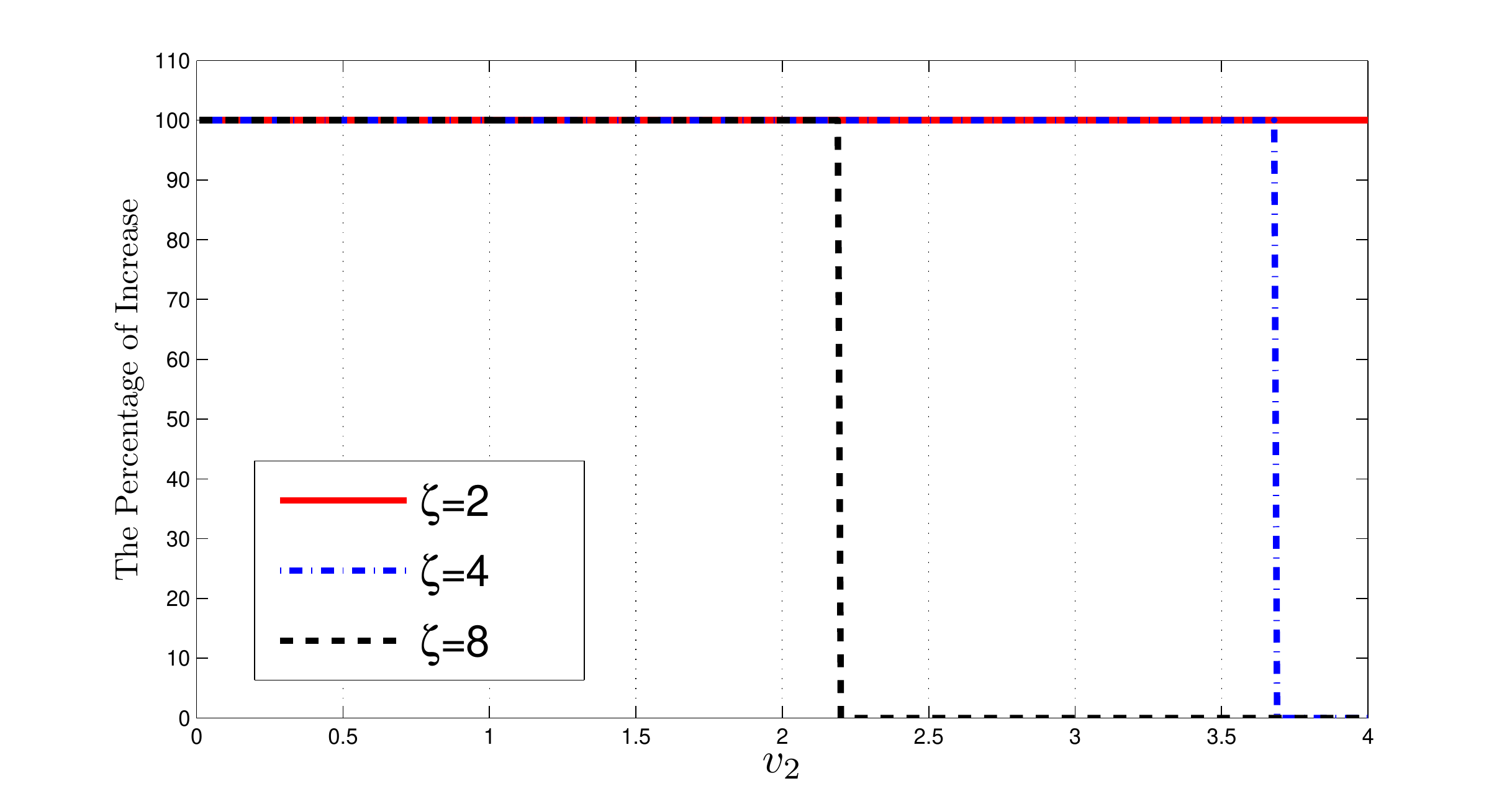}
	\caption{The percentage of increase in the utility of the CP when $\kappa_u =\frac{1}{\zeta}$ with respect to $v_2$ for different values of $\zeta$.}
	\label{figure:increase_CP_equal}
\end{figure}

Note that when the utility before bargaining is zero and the utility after bargaining is positive, then the increase in the utility, using \eqref{equ:increase}, is 100 percent, and this value is zero if the utility after bargaining is  equal to the utility before bargaining.

Results in Figure~\ref{figure:increase_CP_equal} reveal that the percentage of increase in the payoff of the CP is either zero or 100. Note that when $\kappa_u=\frac{1}{\zeta}$, the market either has a stable outcome of no sponsoring or a stable outcome of minimum quality sponsoring. In both cases, the utility of the CP is zero. Thus, if bargaining occurs, the CP would get a positive payoff, and the percentage of the increase in the utility of the CP would be 100. In Figure~\ref{figure:increase_CP_equal}, we can see a threshold on $v_2$ after which the bargaining does not occur. This threshold is decreasing with respect to $\zeta$. The reason is intuitive: even in a bargaining framework, due to limited resources, sponsoring does not occur if the CP needs a high quality to be sponsored, and/or the quality of non-sponsored data is important for the SP.


\begin{figure}[t]
	\centering
	\includegraphics[width=0.4\textwidth]{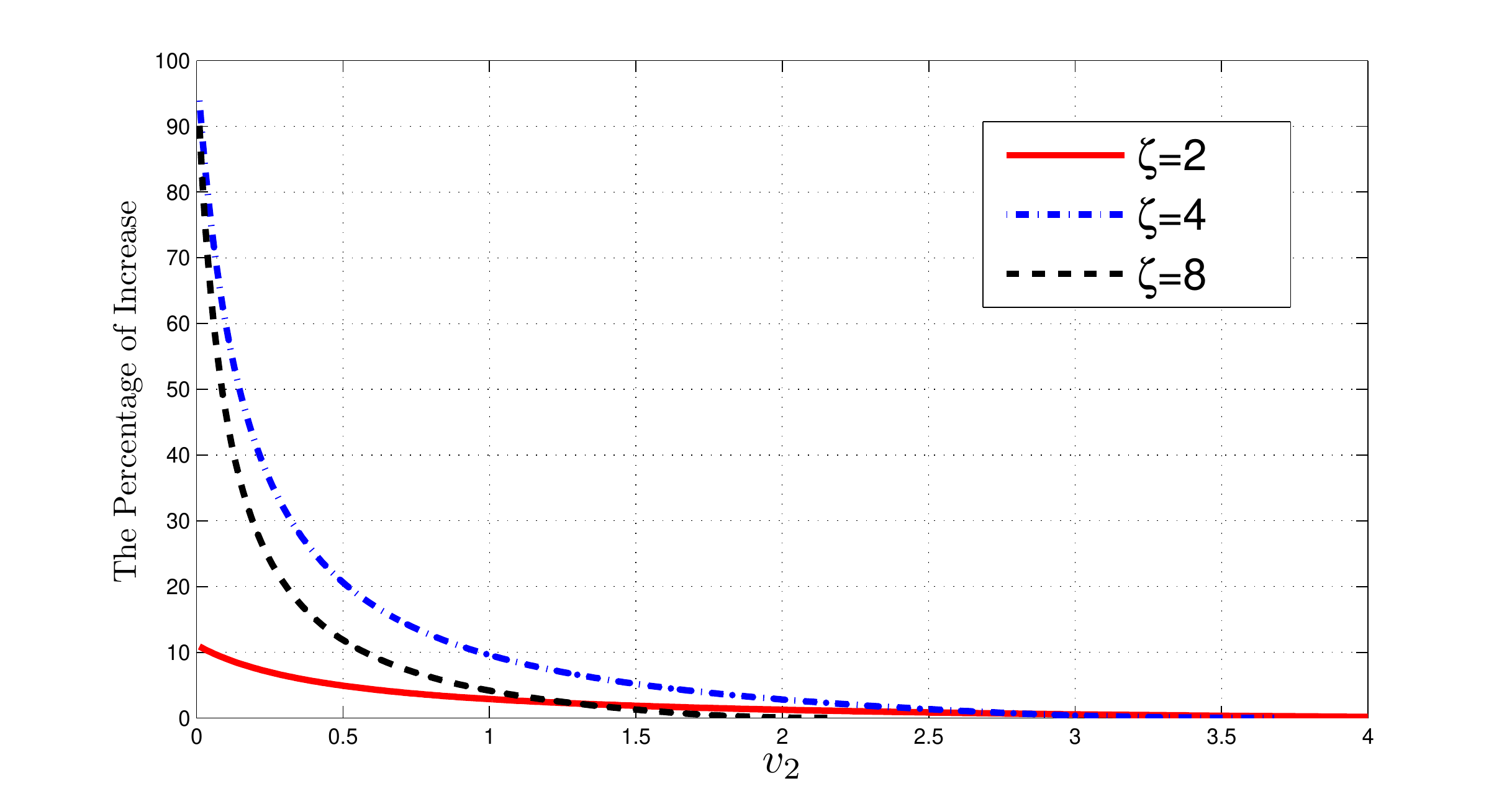}
	\caption{The percentage of increase in the utility of the SP when $\kappa_u =\frac{1}{\zeta}$ with respect to $v_2$ for different values of $\zeta$.}
	\label{figure:increase_SP_equal}
\end{figure}

Results in Figure~\ref{figure:increase_SP_equal} reveal that the percentage of increase in the utility of the SP after bargaining is decreasing with respect to $v_2$. In other words, the higher the importance of non-sponsored data for the end-users and subsequently the SP, the lower the incentive of the SP for participating in a bargaining game.  Note that the case $\zeta=2$ is corresponding to  a minimum quality stable outcome in the short-sighted framework. Thus, bargaining does not add greatly to the utility of the SP. On the other hand, $\zeta=4$ and $\zeta=8$ are corresponding to a stable point of no sponsoring in the short sighted framework. Therefore, the increase in the utility of the SP from bargaining is higher in these two cases than $\zeta=2$. In addition, the percentage of increase is decreasing with respect to $\zeta$. In other words, the higher the minimum quality needed to be sponsored, the lower the incentive of the SP for a bargaining framework.

\begin{figure}[t]
	\centering
	\includegraphics[width=0.4\textwidth]{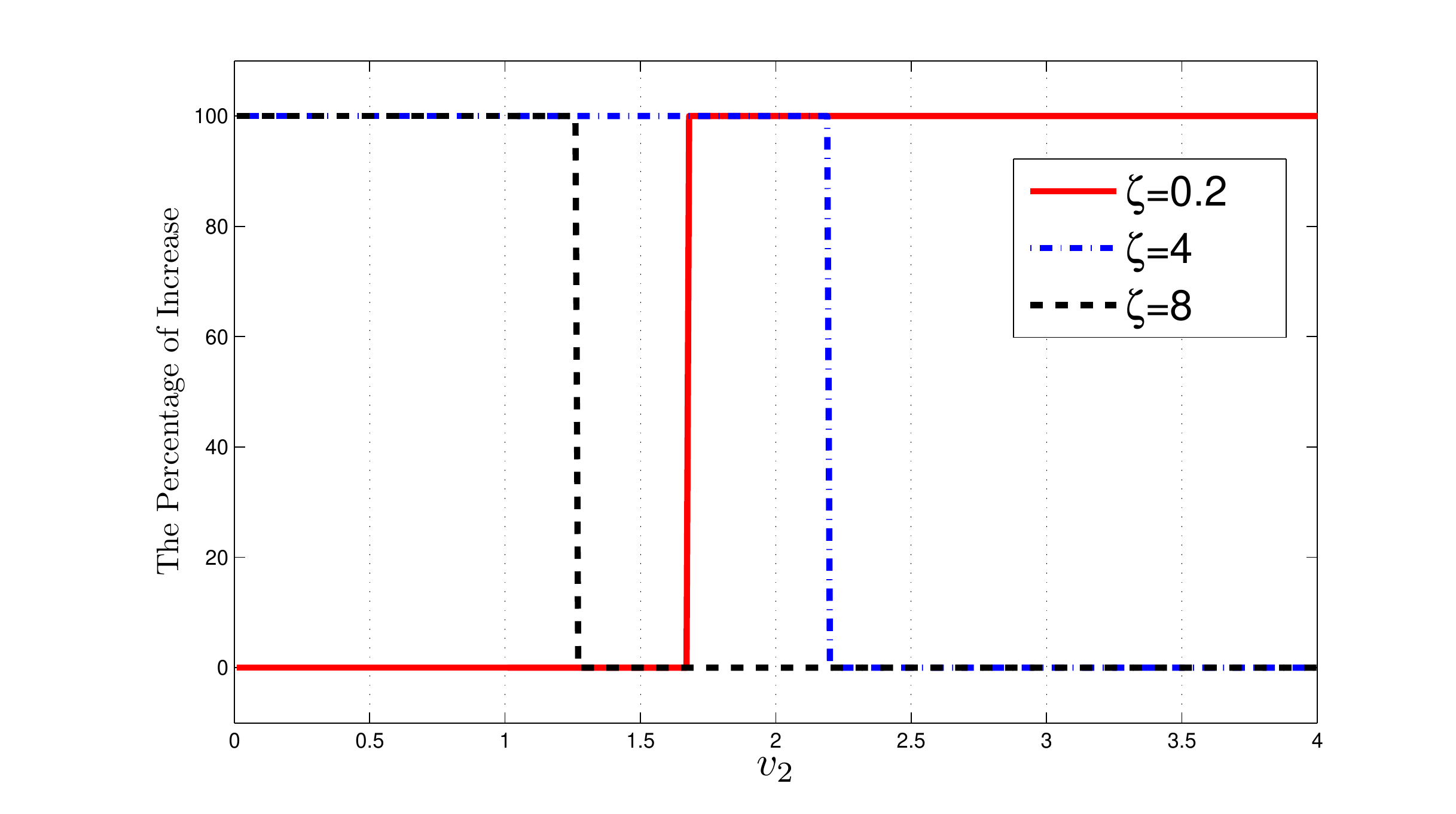}
	\caption{The percentage of increase in the utility of the CP when $\kappa_u =\frac{1}{2\zeta}$ with respect to $v_2$ for different values of $\zeta$.}
	\label{figure:increase_CP_less}
\end{figure}

\begin{figure}[t]
	\centering
	\includegraphics[width=0.4\textwidth]{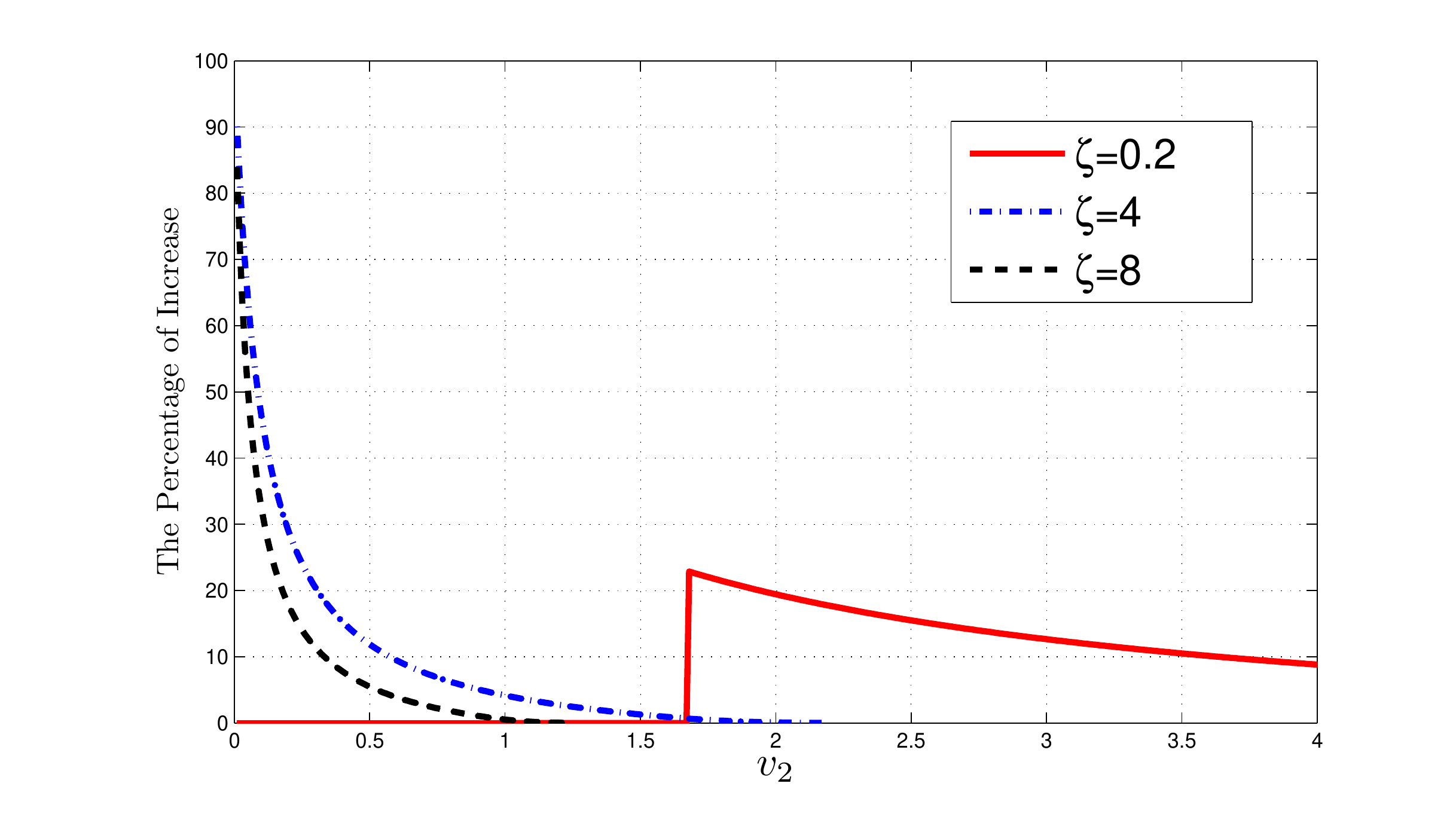}
	\caption{The percentage of increase in the utility of the SP when $\kappa_u =\frac{1}{2\zeta}$ with respect to $v_2$ for different values of $\zeta$..}
	\label{figure:increase_SP_less}
\end{figure}

In Figures~\ref{figure:increase_CP_less} and \ref{figure:increase_SP_less}, the percentage of the increase in the payoff of the CP and the SP is plotted when $\kappa_u=\frac{1}{2\zeta}$. Note that for the case of $\zeta=0.2$, when $v_2$ is small, the stable outcome of a short-sighted market would be the maximum bit sponsoring. Since, in this case, this stable outcome, yields the highest payoff for the SP and the CP, bargaining cannot create additional profit. Thus, the percentage of increase in the utility of the SP and the CP is zero up a threshold. For $v_2$ higher than this threshold, and the cases $\zeta=4$ and $\zeta=8$, the corresponding short-sighted outcome of the market is no sponsoring stable outcome. Thus, the results is similar to the previous figures (Figures \ref{figure:increase_CP_equal} and \ref{figure:increase_SP_equal}). 

Note that bargaining can enforce sponsoring for the set of parameters that have no stable sponsoring outcome in a sequential game. However, the bargaining framework cannot always enforce sponsoring. In particular, if the CP needs to sponsor a high quality (high $\zeta$) for the content, or the quality of non-sponsored content is important for the end-users and subsequently the SP (high $v_2$), then sponsoring does not occur regardless of the framework used.

The next set of numerical results investigate the relation between the number of available resources ($\hat{N}$) and the Nash bargaining price ($p^*$). Intuitively, one may expect that higher  number of available resources yields a lower valuation of the SP for each unit of resources, and subsequently a lower price for each bit. While this line of thought seems to be true in the sequential  framework, numerical results reveal a more complex relationship between $p^*$ and $\hat{N}$ in the bargaining framework: the negotiated price can be increasing, decreasing, or a combination of both (Figures \ref{figure:pvsNhat_less} and \ref{figure:pvsNhat_equal}).

The reason for this counter-intuitive behavior is the different disagreement payoffs resulting from different asymptotic outcomes of the game when decision makers are short-sighted. The disagreement payoffs can be considered as a form of bargaining power for each decision maker, and can affect the excess profit resulted by bargaining. Thus, different disagreement payoffs lead to different amounts of excess profit and its division between the CP and the SP, and subsequently different behavior of price per sponsored bit.


\begin{figure}[t]
	\centering
	\includegraphics[width=0.4\textwidth]{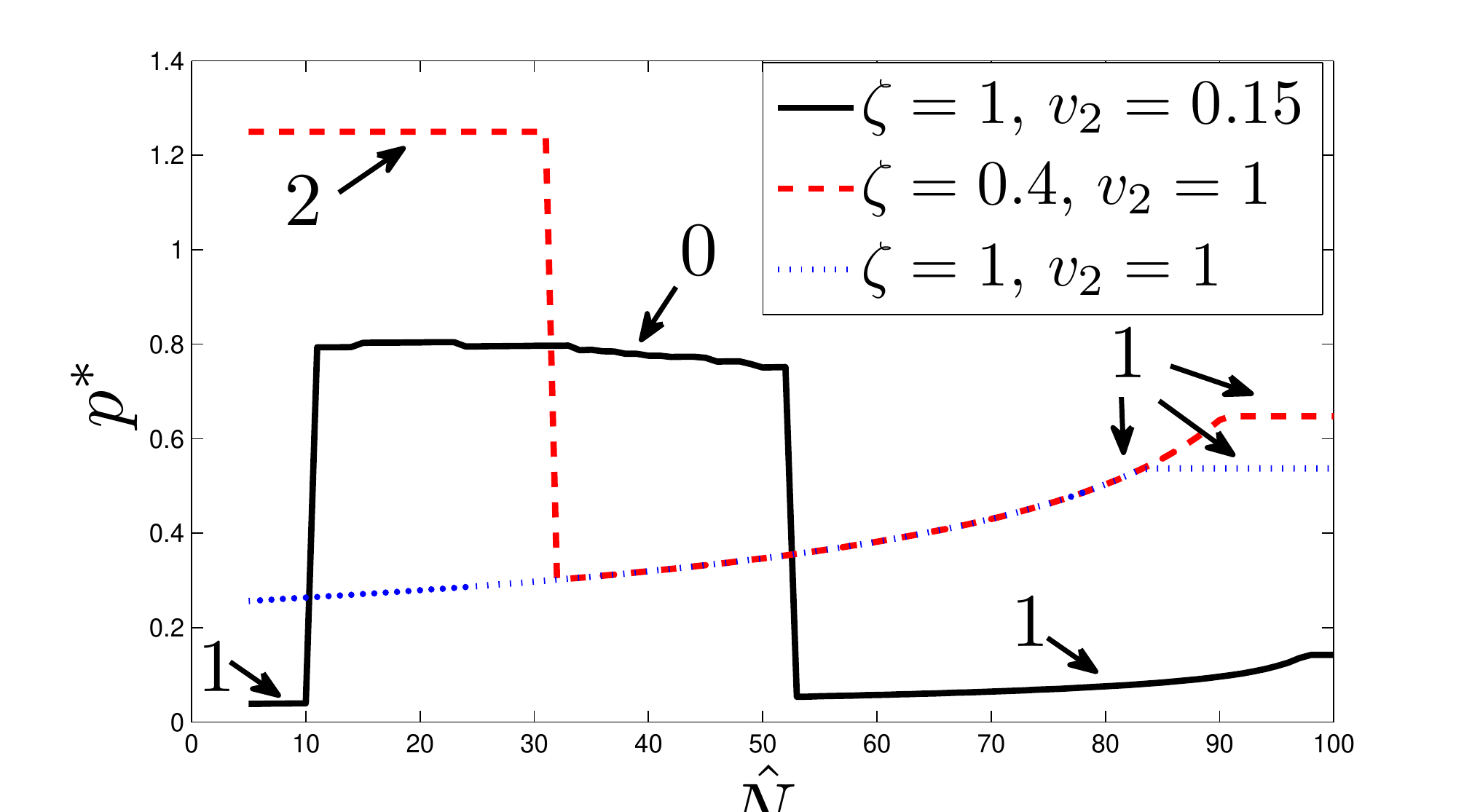}
	\caption{The price per bit in a Nash bargaining solution versus the available number of bits when $\kappa_u=\frac{1}{2\zeta}$.}
	\label{figure:pvsNhat_less}
\end{figure}


\begin{figure}[t]
	\centering
	\includegraphics[width=0.4\textwidth]{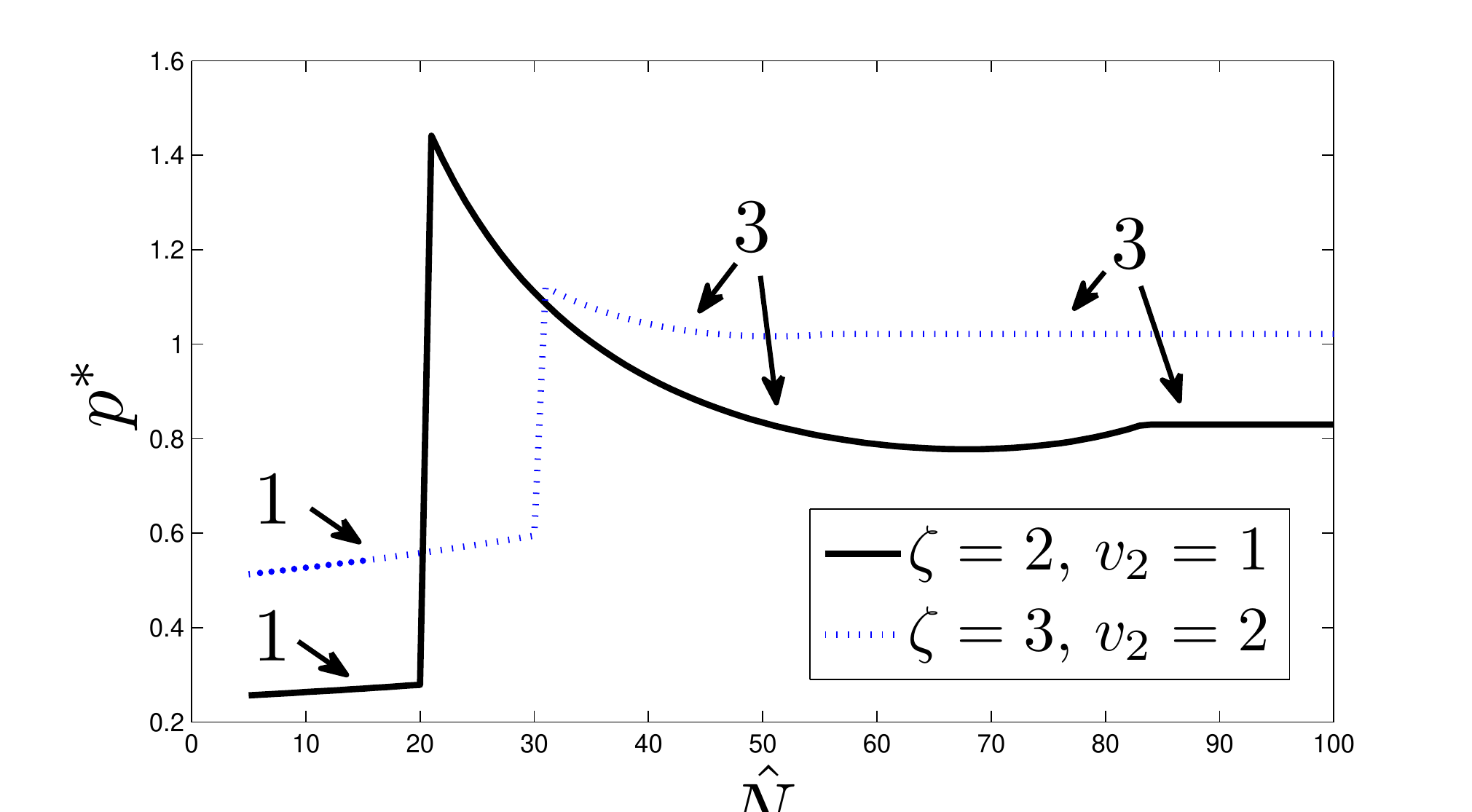}
	\caption{The price per bit in a Nash bargaining solution versus the available number of bits when $\kappa_u=\frac{1}{\zeta}$.}
	\label{figure:pvsNhat_equal}
\end{figure}

\section{Discussions}\label{section:discussion_final}
\edit{First in Section~\ref{section:results}, we present high-level perspective of the  results. Then in Section~\ref{section:modellingAssumptions},  we discuss about the modeling and assumptions of this paper, their implications, and their generalizations.}

\subsection{Summary of Key Results}\label{section:results}
\edit{We discussed that relation between the minimum quality  the CP requests, i.e. $\zeta$, and the stable quality, i.e. $\frac{1}{\kappa_u}$,\footnote{\editt{Recall that the stable quality is defined in Definition~\ref{defintion:stablequality}.}} is an important factor in determining the asymptotic outcome of the market (Section \ref{section:simulation_seq}). In particular if the CP over-provisions the minimum quality, i.e. $\zeta>\frac{1}{\kappa_u}$, then there is no stable sponsoring outcome. \editt{The stability can be achieved when $\zeta<\frac{1}{\kappa_u}$ (under-provision). However, the set of parameters for which the market is stable is larger when $\zeta=\frac{1}{\kappa_u}$. }  Thus, a QSD framework is more likely to emerge  for CPs that know the dynamic of their demand ($\frac{1}{\kappa_u}$) and are willing to disclose it (by requesting $\zeta=\frac{1}{\kappa_u}$).  However, note that in a sequential framework and if $\zeta=\frac{1}{\kappa_u}$, the utility of  the CP would be zero, i.e. the	additional profit of the  CP by sponsoring the content  would be fully extracted if the CP reveals the true value of the stable quality. \editt{Thus, the CP would be indifferent between this scheme and neutrality, while the SP receives a higher payoff in QSD scheme.} While in a bargaining framework (if it happens), \editt{both the CP and the SP receive a higher payoff in comparison to a neutral framework.}  Thus,  a bargaining framework is preferable especially for the CP.}


\edit{We showed that a CP with a volatile demand, i.e. a CP whose users are sensitive to the quality (high $\gamma$), leads to a no-sponsoring outcome in a sequential framework (non-cooperative scenario) if both the SP and CP are short-sighted.  Examples of such CPs are streaming websites and emerging CPs (start ups). Thus, a QSD framework is not a viable scenario in long run for these CPs, if the decision-makers are short-sighted. For these CPs, we can expect a stable QSD framework only if one of the CP or the SP is long-sighted, or in a bargaining game framework.}	\edit{In addition, we showed that even in a bargaining framework, an SP who assigns  high weights to the satisfaction of users that use the non-sponsored data  (high $\frac{v_2}{v_1}$), chooses to not sponsor the content of a CP who needs high quality (high $\zeta$). \footnote{\editt{As explained in the model, the parameter $v_2$ can also represent the regulatory policies for the quality of the experience of the users that use the non-sponsored content. In this case, a high $v_2$ is corresponding to stricter (net-neutrality) rules.}}}

	\edit{Moreover, results reveal that investment by the SP is not always in favor of having a stable QSD. Increasing the number of available resources for sponsoring (investment by the SP),  when at least one of the decision makers is long-sighted, increases the range of the parameters by which a stable QSD framework occurs. However, when both the SP and CP are short-sighted, increasing the number of resources may lead to a scenario in which the CP sponsors a large number of resources. If the demand is volatile, this yields a sudden jump in the demand,  and drives the market to a no-sponsoring outcome (the jump in demand decreases the quality below the minimum quality), and leads the CP to use the best-effort scenario.}

\subsection{Comments on the Assumptions of the Model}\label{section:modellingAssumptions}
\edit{We  assume logarithmic functions  for the demand update function and utilities  owing to its concavity. However, our analysis and insights are expected to be applicable to other concave functions with diminishing returns.}

\edit{Note that the focus of this work is on the interaction between an SP and a CP, and not on the competition among SPs and CPs. In particular, we consider that only one CP wants to sponsor a quality for her users, and the rests stick to the best effort scenario. The effects of other CPs are considered by the SP as part of her utility. Introducing competition among CPs and SPs would introduce another level of strategic decisions by them. It does not necessarily alter the high-level intuitions for the interaction of the CP and the SP provided in this work. For example, we can expect that even under competition, a CP with a volatile demand would not be a good options for a QSD framework in a non-cooperative scenario. However, considering the competition among CPs provides intuitions on the possible structure of the Internet market in future under a QSD framework. For example, a possible outcome would be the case that competitive CPs divide SPs (and subsequently end-users) among themselves and each sponsors the quality of the content on only one of the SPs. Using this, each CP can secure a monopoly over users. This would be a mild version of the ``Internet fragmentation"  which might be an undesirable outcome for users and  from the perspective of the FCC.	A possible direction for future work is to consider the competition over end-users among ISPs and CPs.}

\edit{In addition,  we assume that quality is sponsored by reserving a number of  resources, e.g. LTE time-frames. In general, SPs can sponsor high quality for users of CPs using various methods, e.g. by prioritization of the content of a CP. Analyzing different methods of sponsoring the quality of a content is beyond the scope of this paper.}

\section{Conclusion}
We introduced the problem of quality-sponsored data (QSD) in cellular networks and studied its implications on market entities in sequential and bargaining game frameworks in various scenarios. The direct coupling between the scarce (wireless) resources and the market decisions resulting from QSD has been taken into account, Subgame Perfect Nash Equilibrium and Nash Bargaining Solution of the problem is characterized, and the market dynamics and equilibira have been investigated. We provided strategies for (i) SPs: to determine if and how to price resources, and (ii) CPs: to determine if and how many resources to sponsor (what quality).  
In this work, we focused on the interaction between ISPs and CPs. A possible direction for future work is to consider the competition over end-users among ISPs and CPs. \edit{Another direction is to consider the effects of QSD on the payments of user to SPs, and its implications on the results.
	}

\bibliographystyle{IEEEtran}
\bibliography{bmc_article}

\appendices

\section{Proof of Theorem~\ref{theorem:CP_myopic}}

	First we consider the case in which $d_t > \frac{\hat{N}}{\zeta}$. In this case, 
	$\zeta d_t>\hat{N}$. Therefore, there is no feasible solution for $b_t$. Thus, as we mentioned previously after \eqref{equ:optimization_CP_1}, in this case of infeasibility,  the CP exits the sponsorship program, i.e. $z^*_t=0$. In addition, from \eqref{equation:utilityCP}, $d_t=0$ yields $u_{CP,t}(b_t)<0$ for every $b_t>0$, and subsequently $z^*_t=0$. This completes the proof of \eqref{equ:opt_CP_1_2}. 
	
	Thus, henceforth, we consider $ 0<d_t \leq \frac{\hat{N}}{\zeta}$. Clearly, the utility of the CP \eqref{equation:utilityCP} is concave. Thus, the first order optimality condition provides us with the candidate optimum answer for \eqref{equ:optimization_CP_1}. The first order condition yields that $\hat{b}_t=\frac{\alpha d_{t}}{p_t}$.  In order to be an optimum answer, $\hat{b}_t$ should be feasible, i.e. $\zeta d_t \leq \hat{b}_t \leq \hat{N}$. This characterizes a region for ${p}_t$, $\frac{\alpha d_t}{\hat{N}}\leq p_t\leq \frac{\alpha}{\zeta}$. 
	In order to determine $z^*$, we should check non-negativity of $u^*_{CP,t}$. The utility of the CP  with $\hat{b}_t=\frac{\alpha d_{t}}{p_t}$ is non-negative if $p_t\leq \frac{\alpha \kappa_{CP}}{e}$. Since $\zeta \kappa_{CP}>e$ \footnote{The condition to have a non-trivial problem stated in Section~\ref{section:model}.}, $\frac{\alpha}{\zeta}<\frac{\alpha \kappa_{CP}}{e}$. Therefore, a feasible solution for \eqref{equ:opt_CP_1_1} yields a non-negative payoff. Thus, $z^*_t=1$.  This is the second region from top in \eqref{equ:opt_CP_1_1}.


	If $p_t\leq\frac{\alpha d_t}{\hat{N}}$, then the top boundary condition $b^*_t=\hat{N}$ is the optimum answer of \eqref{equ:opt_CP_1_1}. In addition, since in this region $u_{CP,t}(\hat{N})$ is positive, $z^*_t=1$. This is the first optimality region of \eqref{equ:opt_CP_1_1}.  On the other hand, if $ p_t\geq \frac{\alpha}{\zeta}$, then the lower boundary condition, i.e. $\bar{b}_t=\zeta d_t$, is the optimum answer of the optimization. The condition for  $u_{CP,t}(\bar{b}_t)\geq 0$  and therefore $z^*_t=1$ is  $p_t\leq \frac{\alpha \log(\kappa_{CP} \zeta)}{\zeta}$ which yields the third optimality region in \eqref{equ:opt_CP_1_1}. If $p_t > \frac{\alpha \log(\kappa_{CP} \zeta)}{\zeta}$, $u_{CP,t}(b_t)<0$. Thus, $z^*_t=0$. This concludes the proof.

	\section{Proof of Theorem~\ref{Theorem:SP}}

	Theorem~\ref{theorem:CP_myopic} implies that if $ d_t> \frac{\hat{N}}{\zeta}$, or $d_t=0$, or $p_t>\frac{\alpha \log(\kappa_{CP}\zeta)}{\zeta}$, the CP does not participate in the sponsoring program. Thus, the value of $y^*_t$ does not affect the outcome of the market in these cases. Without loss of generality, we assume that in these cases the SP does not offer the program, i.e. $y^*_t=0$.

	Thus, henceforth, we consider $0<d_t\leq \frac{\hat{N}}{\zeta}$ and $p_t\leq \frac{\alpha \log(\kappa_{CP}\zeta)}{\zeta}$. Note that in this region, by Theorem~\ref{theorem:CP_myopic}, $b_t>0$. Thus, the SP maximization problem is,
	
	\small
	\begin{equation}
	\ba 
	&\max_{p_t} u_{SP,t}(p_t)=\max_{p_t} \Bigg{(} p_t b^*_t + \nu_1   d_{t} \log\l \frac{\kappa_{SP} b^*_t}{d_{t}}\r+\\
	&\qquad \qquad \qquad \qquad \qquad  +\nu_2 D \log \l \kappa_{SP} \frac{N-b^*_t}{D}\r \Bigg{)},
	\ea
	\end{equation}
	\normalsize
	where $b^*_t$ is the equilibrium outcome of the second stage. Let $p_t\leq \frac{\alpha d_t}{\hat{N}}$. Then from Theorem~\ref{theorem:CP_myopic}, $b^*_t=\hat{N}$. Thus, $u_{SP,t}(p_t)$ is a strictly increasing function of $p_t$. Therefore,  all prices less than $\frac{\alpha d_t}{\hat{N}}$ yields a strictly lower payoff than $p^*_{1,t}=\frac{\alpha d_t}{\hat{N}}$, which is the first candidate pricing strategy. Next, let  $\frac{\alpha}{\zeta}\leq p_t \leq \frac{\alpha \log(\kappa_{CP} \zeta)}{\zeta}$. Thus, from Theorem~\ref{theorem:CP_myopic}, $b^*_t=\zeta d_t$. Again, in this region, $u_{SP,t}(p_t)$ is a strictly increasing function of $p_t$. Thus, $p^*_{2,t}=\frac{\alpha \log(\kappa_{CP} \zeta)}{\zeta}$ strictly dominates all other prices in this interval, which yields the second candidate pricing strategy\footnote{ Note that $p^*_{2,t}=\frac{\alpha \log(\kappa_{CP} \zeta)}{\zeta}$ yields a payoff of zero for the CP. However, since we have assumed that the indifferent CP chooses to join the sponsorship program, $z^*_t=1$ and subsequently $y^*_t=1$.}. For the case that $\frac{\alpha d_t}{\hat{N}}\leq p_t\leq \frac{\alpha}{\zeta}$, from Theorem~\ref{theorem:CP_myopic}, $b^*_t=\frac{\alpha d_t}{p_t}$. In this region, the first order condition on $u_{SP,t}(p_t)$ provides us with the local extremum,
	
	\small
	\begin{equation}\label{equ:opt_p_2}
	p^*_{3,t}=\alpha \frac{\nu_1d_t+\nu_2 D}{\nu_1 N}
	\end{equation}
	\normalsize
	Since the second order derivative can be negative or positive, the first order condition provides us with only a candidate optimum answer,
	which is the third candidate pricing strategy. This candidate strategy should satisfy the condition $\frac{\alpha d_t}{\hat{N}}\leq p^*_{3,t}\leq \frac{\alpha}{\zeta}$. If not, it would not be an optimum answer since, as we discussed earlier in the proof, every price less than (respectively, higher than) $\frac{\alpha d_t}{\hat{N}}$ (respectively, $\frac{\alpha}{\zeta}$ ) is dominated by $\frac{\alpha d_t}{\hat{N}}$ (respectively, $\frac{\alpha \log(\kappa_{CP} \zeta)}{\zeta}$)\footnote{Note that from Theorem~\ref{theorem:CP_myopic}, prices higher than $\frac{\alpha \log(\kappa_{CP} \zeta)}{\zeta}$ leads to no sponsoring on the CP side.}. Note that these candidate strategies are optimum only if they yield a payoff higher than the payoff of the SP in the case of no-sponsoring, i.e. $v_2 D \log(\kappa_{SP}\frac{N}{D})$. The result follows.

\section{Proof of Theorem~\ref{theorem:stable}}

	We characterize the possible stable outcomes of the game when at each time $t$, the SP and the CP choose their strategy to be the SPNE of the game characterized in Theorems~\ref{theorem:CP_myopic} and \ref{Theorem:SP}.

	The first candidate stable outcome is trivial: as soon as one of the CP or SP exits the sponsorship program, or $d_t>\frac{\hat{N}}{\zeta}$, or $d_t=0$, the program will not be resumed.

	Now consider the case that sponsoring occurs. In this case, $y=1$ $z=1$, and from Theorem~\ref{Theorem:SP}, the SP chooses one of the candidate optimum pricing strategies from the set $P^*=\{\frac{\alpha d_t}{\hat{N}},\frac{\alpha \log \l \kappa_{CP}\zeta \r}{\zeta},\alpha \frac{\nu_1d_t+\nu_2 D}{\nu_1 N}\}$.  We show that the first, the second, and the third candidate pricing strategies are corresponding to the second, the third, and the fourth stable outcome, respectively. Note that when choosing these prices, by Theorem~\ref{Theorem:SP}, the demand should be feasible, i.e. $0<d_t\leq \frac{\hat{N}}{\zeta}$. In addition, recall that by Lemma~\ref{corollary:quality_24}, the demand is stable when $d=\kappa_{u}b$.

	Now, we obtain the second stable outcome by considering that $p=\frac{\alpha d}{\hat{N}}$ and $0<d_t\leq \frac{\hat{N}}{\zeta}$.  In this case, from Theorem~\ref{theorem:CP_myopic}, $b=\hat{N}$. Thus  $p=\alpha \kappa_u$ since $d=\kappa_u \hat{N}$. The feasibility condition yields that $d=\kappa_u \hat{N}\leq \frac{\hat{N}}{\zeta}\Rightarrow \kappa_u\leq \frac{1}{\zeta}$ \footnote{Note that $d=\kappa_u \hat{N}>0$.}.

	Next, we obtain the third stable outcome by considering $p=\frac{\alpha \log(\kappa_{CP} \zeta)}{\zeta}$ and $0<d_t\leq \frac{\hat{N}}{\zeta}$. 
	In this case, from Theorem~\ref{theorem:CP_myopic}, $b=\zeta d$, and subsequently from the stability condition, $d=\kappa_u b=\kappa_u \zeta d$. Therefore, this case occurs if  $\kappa_u \zeta=1$. Note that the demand  could be any positive value less than or equal to $\frac{\hat{N}}{\zeta}$ (feasibility condition), and with this demand, $0<b=\zeta d \leq \hat{N}$.

	Finally, the fourth possible stable outcome happens when $p=\alpha \frac{\nu_1d+\nu_2 D}{\nu_1 N}$, $p\in[\frac{\alpha d}{\hat{N}},\frac{\alpha}{\zeta}]$ (from Theorem~\ref{Theorem:SP}), and  $0<d_t\leq \frac{\hat{N}}{\zeta}$. In this case, from Theorem~\ref{theorem:CP_myopic}, $b=\frac{\alpha d}{p}$. In order to have a stable outcome,  $d=\kappa_u b \Rightarrow p=\alpha \kappa_u$. Thus, from $p=\alpha \frac{\nu_1d_t+\nu_2 D}{\nu_1 N}$, $d=N \kappa_u - \frac{\nu_2}{\nu_1}D$ and $b=N - \frac{\nu_2}{\nu_1\kappa_u}D$. Note that $b$ should satisfy $0<b\leq \hat{N}$, and from Theorem~\ref{Theorem:SP}, we know that  $p=\alpha \frac{\nu_1d+\nu_2 D}{\nu_1 N}$ is optimum if it is in the interval $[\frac{\alpha d}{\hat{N}},\frac{\alpha}{\zeta}]$. The latter yields that $\frac{\alpha d}{\hat{N}}\leq p=\alpha \kappa_u \leq \frac{\alpha}{\zeta}$, which yields that $\kappa_u\leq \frac{1}{\zeta}$ and $b=\frac{\alpha d}{p}\leq \hat{N}$. Note that these conditions automatically lead to a feasible demand: from $b=\frac{\alpha d}{p}\leq \hat{N}$, then $d\leq \frac{\hat{N}p}{\alpha}=\hat{N}\kappa_u\leq \frac{\hat{N}}{\zeta}$. Thus, in this stable outcome, $\kappa_u\leq \frac{1}{\zeta}$ and $0<b\leq \hat{N}$.
	The result follows.

\section{Proof of Theorem~\ref{theorem:SPd}}
	By \eqref{equ:SPutility}, the utility of the SP when choosing the tuple  $\l d,1,\frac{\alpha \log\l \kappa_{CP}\zeta\r}{\zeta},1,\zeta d\r$ is:
	\small
	\begin{equation}
	u_{SP}=\alpha d \log\l \kappa_{SP} \zeta \r +\nu_1 d \log\l \kappa_{SP} \zeta\r + \nu_2 D \log \l \kappa_{SP} \frac{N-\zeta d}{D} \r
	\nonumber
	\end{equation} 
	\normalsize
	
	First, note that the expression of the utility is concave in $d$. Thus, the first order condition gives the optimum answer. The solution of the first order condition is: 
	\begin{equation}
	d^*=\frac{N}{\zeta}-\frac{1}{\l \alpha+\nu_1\r \log\l \kappa_{SP}\zeta\r}
	\nonumber
	\end{equation}
	Based on Theorem~\ref{theorem:stable}, for $d^*$ to be the demand corresponding to the minimum quality stable outcome, it should satisfy the constraint  $0<d^*\leq \frac{\hat{N}}{\zeta}$. If $d^*>\frac{\hat{N}}{\zeta}$ or $d^*<0$, the concavity implies that the optimum is $d=\frac{\hat{N}}{\zeta}$ or $d=0$, respectively.

\section{Proof of Theorem~\ref{Theorem:longrunSP}}
	First, note that in Theorem~\ref{theorem:stable}, when $\kappa_u=\frac{1}{\zeta}$, the stable points 2, 3, and 4 can occur. In addition, the demand is fixed in the stable point 2 and 4, while it can take a range of values for the stable point 3, including the fixed demands in the other two stable points. On the other hand, the price is fixed in all these three stable points. In these stable points, the stable quality is $\frac{b}{d}=\zeta$. Thus, by \eqref{equ:SPutility}, the payoff of the SP is:
	$$
	u_{SP} =  p \zeta d+ \nu_1   d \log\l \kappa_{SP} \zeta \r+\nu_2 D \log \l \kappa_{SP} \frac{N-\zeta d}{D}\r 
	$$
	\normalsize
	Therefore, for a fixed demand, the payoff of the SP in this case is an increasing function of the price $p$. Note that the third stable point, i.e. minimum quality stable point, has the highest price among the possible stable points since $\log(\kappa_{CP} \zeta)>1$.\footnote{In Section~\ref{section:model}, we assumed that in order to have a non-trivial problem $\kappa_{CP}\zeta>e$}    In addition, it can take a range of demand including the fixed demands of the stable point 2 and 4 . Thus, the third stable outcome of the market yields the highest payoff for the SP. The optimum demand is chosen by Theorem~\ref{theorem:SPd} as discussed before. The result follows.

\section{Proof of Theorem~\ref{theorem:CP_long_sighted}}

First note that from Corollary~\ref{remark:indifferent}, the minimum quality stable point, i.e. $\l d,1,\frac{\alpha \log\l \kappa_{CP}\zeta\r}{\zeta},1,\zeta d\r$, yields a payoff of zero for the CP. From Lemma~\ref{corollary:quality_24}, in both the maximum bit sponsorship, i.e. $\l \kappa_u \hat{N}, 1, \alpha \kappa_u, 1, \hat{N}\r$, and  interior stable point, i.e. $\l N \kappa_u - \frac{\nu_2}{\nu_1}D, 1, \alpha \kappa_u,1, N - \frac{\nu_2}{\kappa_u \nu_1}D\r$,  the stable quality ($\frac{b}{d}$) is $\frac{1}{\kappa_u}$. Thus, using \eqref{equation:utilityCP}, the payoff of the CP in these plausible stable outcomes is:
\begin{equation}
u_{CP}=\alpha d \l \log\l \frac{\kappa_{CP}}{\kappa_{u}}\r -1\r
\nonumber
\end{equation}
Note that from the condition for plausibility of these stable points ($\kappa_u\leq \frac{1}{\zeta}$), and our previous assumption that $\kappa_{CP}\zeta >e$\footnote{The condition to have a non-trivial problem.}, $\frac{\kappa_{CP}}{\kappa_{u}}>e$. Thus, the payoff of the  CP in the maximum bit sponsorship and the interior stable point is strictly greater than zero, and is strictly increasing with respect to the demand. Given that the quality $\frac{b}{d}=\frac{1}{\kappa_u},$ and is a constant, the higher the number of sponsored bits, the higher the demand, and therefore the higher the payoff of the CP would be. In addition, note that the number of sponsored bits in the maximum bit sponsorship point is greater than or equal to the number of sponsored bits in the interior stable point. Thus, the utility of the CP in the  maximum bit sponsoring point is greater than or equal to the utility in the interior stable point. The result follows.

\section{Proof of Theorem~\ref{theorem:bargaining}}
	$d_{CP}$ and $d_{SP}$ are independent of $d$ and $p$. In addition, $u_{excess}=(u_{CP}-d_{CP})+(u_{SP}-d_{SP})$ is independent of $p$, and is only a function of $d$. Thus, for a given $d$, using equation (2) of \cite{bargainingsplit}, the optimum value of $p$ is such that:
	\begin{equation}\label{equ:optimum}
	\frac{u_{CP}-d_{CP}}{w}=\frac{u_{SP}-d_{SP}}{1-w}
	\nonumber
	\end{equation}
	if the solution for $p$ satisfies other constraints. Thus, by plugging the expressions for the CP and the SP (\eqref{equ:uCPbarg} and \eqref{equ:uSPbarg}), the candidate optimum $p$ as a function of $d$ is:
	\begin{equation}\label{equ:optimum_p_2}
	\begin{aligned}
	p^*&=\frac{\kappa_u}{d} \Big{(}(1-w)(u_{ad}-d_{CP})-w(u_s-d_{SP})\Big{)}\\
	&=\frac{\kappa_u}{d} \Big{(}(u_{ad}-d_{CP})-wu_{excess}\Big{)}
	\end{aligned}
	\end{equation}
	
	Substituting \eqref{equ:optimum_p_2} in the objective function of \eqref{equ:nash_solution_2} and using \eqref{equ:uCPbarg} and \eqref{equ:uSPbarg} yield the new objective function:
	\begin{equation}
	w^w (1-w)^w (u_{ad}-d_{CP}+u_s-d_{SP})=w^w (1-w)^w u_{excess}
	\nonumber
	\end{equation}
	\normalsize
	Substituting \eqref{equ:optimum_p_2}, \eqref{equ:uCPbarg}, and \eqref{equ:uSPbarg}  in the constraint $u_{CP}\geq d_{CP}$, yields the new constraint $u_{ad}-d_{CP}+u_s-d_{SP}\geq 0$. Similar substitutions for $u_{SP}\geq d_{SP}$ yields the same constraint. Thus, the optimization can be written as,
	\begin{equation}
	\small 
	\begin{aligned}
	&\max_{d} u_{excess}\\
	& \text{s.t.}\\
	&\qquad 0\leq d\leq \hat{N}\kappa_u \\
	&u_{ad}-d_{CP}+u_s-d_{SP}=u_{excess}\geq 0
	\end{aligned}
	\end{equation}
	\normalsize
	The theorem follows from above and \eqref{equ:optimum_p_2}.

%
%
%


%





\ifCLASSOPTIONcaptionsoff
  \newpage
\fi

\newpage
\section{Comments on the Approximations in the Model}\label{appendix:approx}

Note that in our model, we have assumed that either the CP sponsor a quality for her end-users or she uses the best effort scenario (both cannot happen together). This means that in the second case (no sponsoring) the demand of the CP would be added to the pool of the demand for the best effort scenario, i.e. would be added to $D$. In our model, we do not considered the augmentation since we naturally expect the demand for a CP to be much smaller than the total demand for all CPs. In this section, we discuss if and how the results change if we consider this augmentation.

\subsection{Change in the Model}
The augmentation in the demand can be accommodated as follows:

\subsubsection{The SP}
\small
\begin{equation}\label{equ:SPutility}
u_{SP,t}(p_t) =  p_t b_t+ u_s(b_t(p_t))
\end{equation}
\normalsize 
where now the users' satisfaction function, i.e. $u_s(.)$, becomes:
\footnotesize
\begin{subnumcases}{u_s(b_t)=}
  \nu_1   d_{t} \log\l \frac{\kappa_{SP} b_t}{d_{t}}\r+\nu_2 D \log \l \kappa_{SP} \frac{N-b_t}{D}\r \label{equ:new_A} \\
   \qquad \qquad \qquad \qquad \qquad \qquad \qquad  \ \  \mbox{if }  d_t>0 \& \ b_t>0 \nonumber \\
 	\nu_2 (D+d_t) \log \l \kappa_{SP} \frac{N-b_t}{(D+d_t)}\r  \quad \mbox{Otherwise} \label{equ:new_B}
\end{subnumcases}


\normalsize

Note that \eqref{equ:new_A} is the same as \eqref{equ:old_A}. Thus, the only change is for the case of no sponsoring ($d_t=0$ or $b_t=0$) \eqref{equ:new_B} in which $d_t$ is added to the total demand of the best effort scenario, i.e. $D$. Note that \eqref{equ:new_B} becomes similar to \eqref{equ:old_B} when $d_t<<D$. 

\subsubsection{The CP} Note that we have considered that the CP receives a payoff of zero in the case of no sponsoring. This is justified as in many cases in which if the CP does not sponsor the data, then she will only transmit the content with a best effort scenario and because of limited bandwidth do not transmit advertisements.  An example of this can be seen in Youtube: If the quality of the content is low, then Youtube automatically skips the ad.  Thus, in this case, when the CP transmits with best effort, it receives zero ad revenue.

\subsection{Change in the Analytical Results}\label{appendix:approx_simulation}
This change may only affect the results when (i) the exact expression of $u_s(b_t)$ in the case of no sponsoring, i.e. \eqref{equ:new_B}, or (ii) the expressions for the optimum strategies of the SP, i.e. $p^*_t$ is used. Note that in Theorem 1 we do not use any of (i) and (ii). Thus, the Theorem would be similar to before.  In the next paragraph, we will argue that the expressions for $p^*_t$ in Theorem 2 would be the same as before. In Theorem 3, we only use the expression for  $p^*_t$. Thus, the results of Theorem 3 would be the same as before. In Theorems 4 and 5, we use \eqref{equ:new_A} (which is similar to \eqref{equ:old_A}) and the expression of $p^*_t$ (which are the same as before). Thus, the the results for these theorems also would be as before. For the long-sighted case, we do not use the exact expression of $u_s(b)$. Thus, all the results of long-sighted would be as before.


Now, we argue that the expressions for the optimum strategies of the SP in Theorem 2 would be the same as before.  The first paragraph of the proof would be the same as before since we do not use \eqref{equ:new_B}. In addition, in the next paragraph of the proof and when characterizing the optimum strategies of the SP, we focus on $0<d_t\leq \frac{\hat{N}}{\zeta}$ and $p_t\leq \frac{\alpha \log(\kappa_{CP}\zeta)}{\zeta}$. With these conditions, $b_t>0$ and sponsoring occurs. Thus, we use the expression of $u_{SP}(p_t)$ for the case of sponsoring \eqref{equ:new_A} which is the same as before, i.e. \eqref{equ:old_A}. Thus, the expressions for the optimum strategies of the SP would be the same as before.  

Note that the only change that should be applied to Theorem~\ref{Theorem:SP} is to the expression of $u_{SP,0}$, i.e. the utility of the SP in the case of no sponsoring. This utility should be changed from \eqref{equ:old_B} to \eqref{equ:new_B}, i.e. $u_{SP,0}=v_2(D+d_t)\log(\kappa_{SP}\frac{N}{D+d_t})$.

\subsection{Change in the Simulation Results} Since the SP now receives a greater utility in the case of no sponsoring (compare \eqref{equ:old_B} with \eqref{equ:new_B}), the option of no-sponsoring becomes more attractive for the SP.  We have redone all the simulations with the new model.  We comment on all the changes in the numerical results, and present the results for one sample scenario (Figure~\ref{figure:increase_SP_equal_newmodel}). 

\textbf{Changes to Figures 3 to 6:}  In the numerical results, we observe that these figures will remain similar  in general.  The only change is that  the region of no sponsoring for large $v_2$ slightly increases (since the no-sponsoring is now more attractive for the SP). Thus, the insights associated with these figures would be the same as before. 

\textbf{Changes to Figures 7 and 9:} Now, consider the numerical results for the long-sighted scenario. In this case, for Figures 7 and 9, we observe the thresholds for the jump to no-sponsoring region slightly decreases (as we expect because of the explanations in the first paragraph of Appendix~\ref{appendix:approx_simulation}). Otherwise, the figures would be the same as before. This is because of the fact that the utility of the CP is the same as before.

\textbf{Changes to Figures 8 and 10:} We plot the counterpart of Figure 8, in Figure~\ref{figure:increase_SP_equal_newmodel}. Note that the results are similar. The only difference is that the percentage of increase in the utility of the SP decreases in some regions (regions in which short-sighted yields no sponsoring). This is because of the increase in the utility of the SP in the case of no sponsoring. The same happens to Figure 9. Thus, the insights associated with these figures remain the same.

	\begin{figure}[t]
		\centering
		\includegraphics[width=0.4\textwidth]{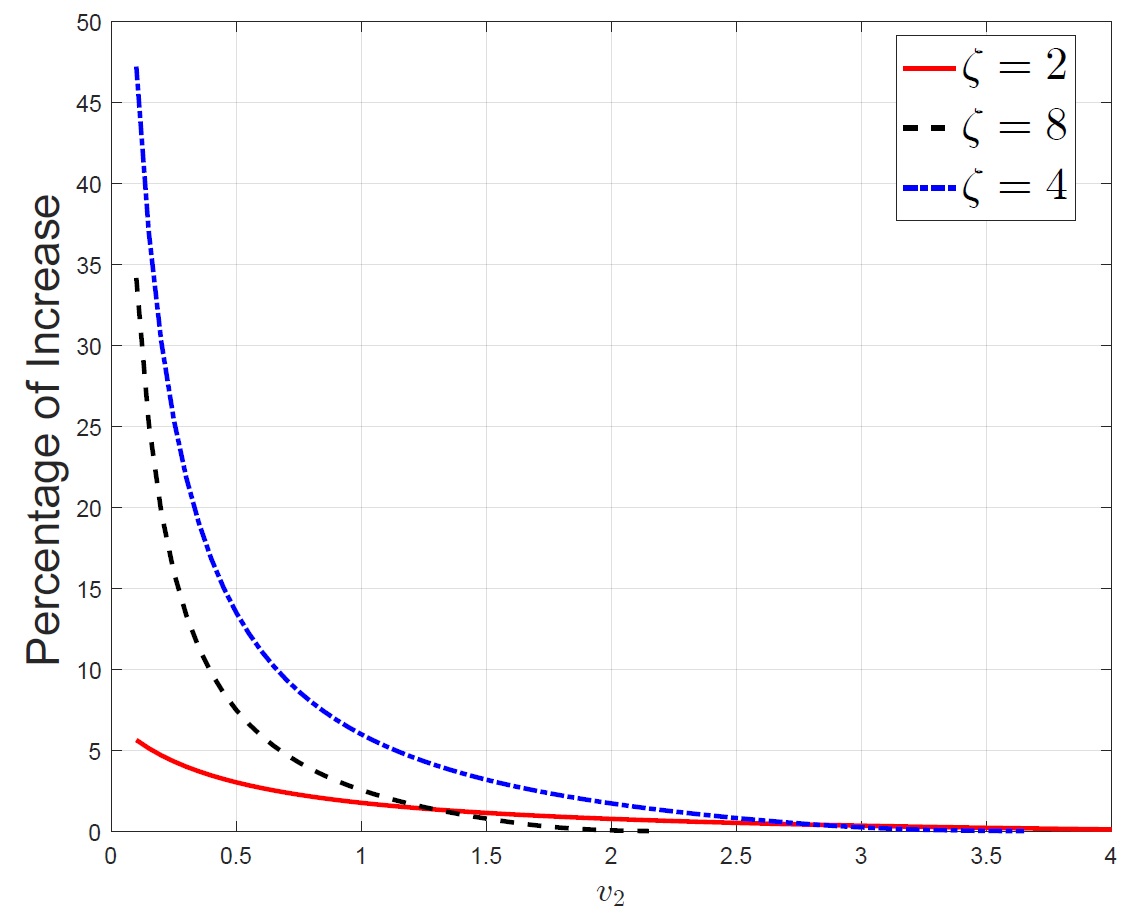}
		\caption{The percentage of increase in the utility of the SP when $\kappa_u =\frac{1}{\zeta}$ with respect to $v_2$ for different values of $\zeta$ (new model).}
		\label{figure:increase_SP_equal_newmodel}
		\end{figure}
		
			\begin{figure}[t]
		\centering
		\includegraphics[width=0.4\textwidth]{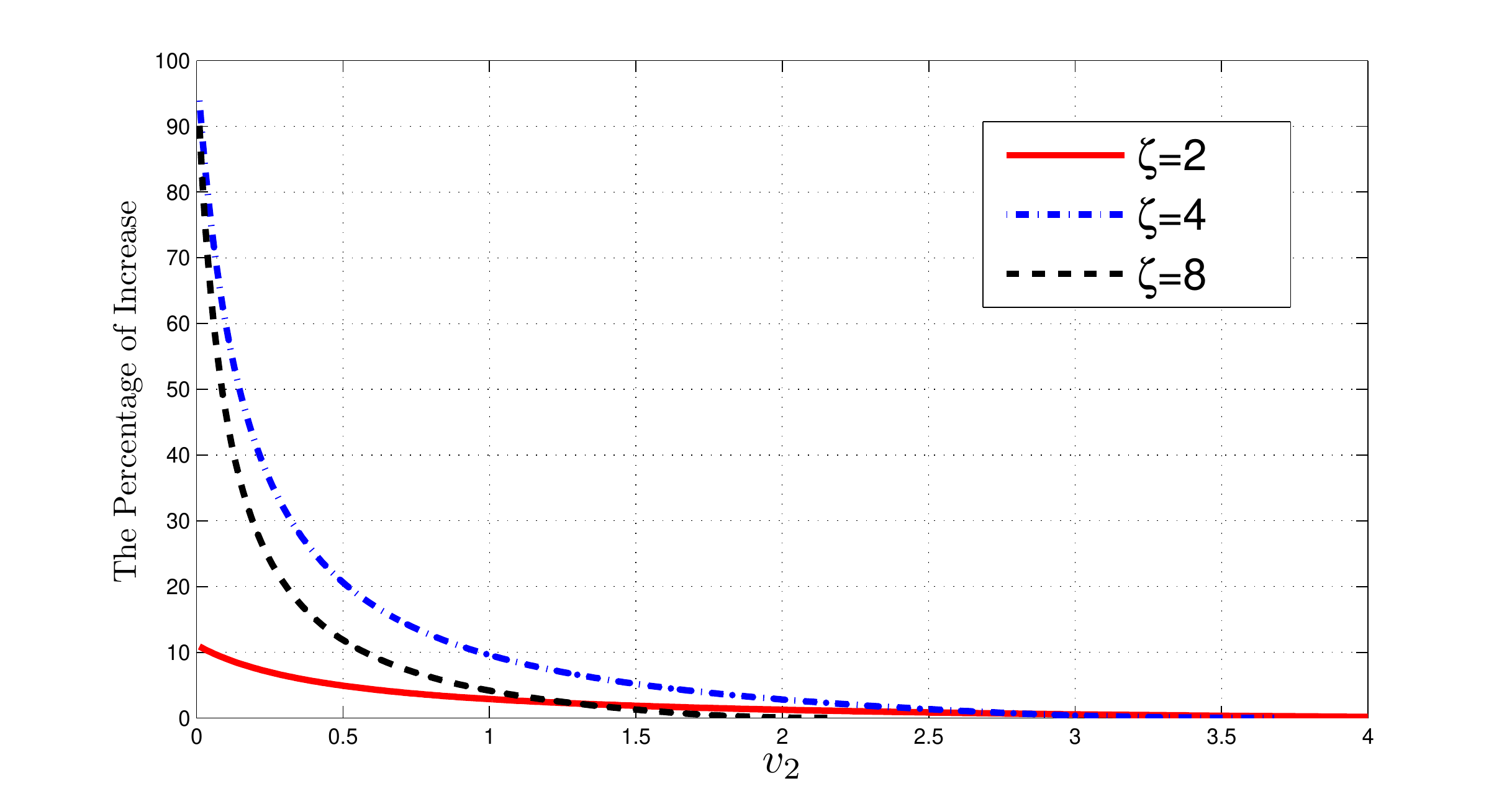}
		\caption{The percentage of increase in the utility of the SP when $\kappa_u =\frac{1}{\zeta}$ with respect to $v_2$ for different values of $\zeta$ (old model).}
		\end{figure}
		
	\textbf{Changes to Figures 11 and 12:}	Recall that $p^*$ is the price of sponsored bits in the bargaining framework, and is distinct from $p^*_t$ which is the price of sponsored bits in the short-sighted framework. Results reveal that the insights associated with these figures follow the same trend as before. Note that $p^*$ depends on the disagreement payoff which is the payoff of short-sighted framework. Thus, the only change to the value of $p^*$ happens when the disagreement yields no sponsoring. In this case, since the payoff of disagreement increases slightly \eqref{equ:new_B}, $p^*$ increases slightly.

\end{document}